%% file: arxiv.tex
\documentclass[letterpaper]{article} % DO NOT CHANGE THIS
\usepackage{aaai2026}  % DO NOT CHANGE THIS
\usepackage{times}  % DO NOT CHANGE THIS
\usepackage{helvet}  % DO NOT CHANGE THIS
\usepackage{courier}  % DO NOT CHANGE THIS
\usepackage[hyphens]{url}  % DO NOT CHANGE THIS
\usepackage{graphicx} % DO NOT CHANGE THIS
\usepackage{natbib}  % DO NOT CHANGE THIS AND DO NOT ADD ANY OPTIONS TO IT
\usepackage{caption} % DO NOT CHANGE THIS AND DO NOT ADD ANY OPTIONS TO IT
\usepackage{algorithm}
\usepackage{algpseudocode}

\usepackage{amsfonts}
\usepackage{amsmath}
\usepackage{amsthm}
\usepackage{relsize}
\usepackage{units}
\usepackage{mathtools}
\usepackage{thm-restate}

\usepackage{booktabs}
\usepackage{makecell}
\usepackage{array}

\usepackage{array}
\renewcommand*{\arraystretch}{1.1}
\newcommand*{\mline}[1]{%
\begingroup
    \renewcommand*{\arraystretch}{1}%
   \begin{tabular}[c]{@{}>{\raggedright\arraybackslash}p{2cm}@{}}#1\end{tabular}%
  \endgroup
}

\newtheorem{theorem}{Theorem}[section]
\newtheorem{corollary}{Corollary}[section]
\newtheorem{lemma}{Lemma}[section]
\newtheorem{definition}{Definition}[section]

\newcommand{\picturewidth}{0.85 \linewidth}

\usepackage{newfloat}
\usepackage{listings}
\DeclareCaptionStyle{ruled}{labelfont=normalfont,labelsep=colon,strut=off} % DO NOT CHANGE THIS
\floatstyle{ruled}
\newfloat{listing}{tb}{lst}{}
\floatname{listing}{Listing}
\title{Fairly Dividing Non-identical Random Items: Just Sample or Match}
\author{
    Aprup Kale\textsuperscript{\rm 1},
    Rucha Kulkarni\textsuperscript{\rm 1},
    Navya Garg\textsuperscript{\rm 2}\\
}
\affiliations {
    \textsuperscript{\rm 1}National University of Singapore,\\
    \textsuperscript{\rm 2}Indian Institute of Technology, Bombay\\
    aprup.kale@u.nus.edu, 
    rucharpk@nus.edu.sg, 
    23b0982@iitb.ac.in
}

\begin{document}

\maketitle

\begin{abstract}
% START -----------------------------------------------------------
We study the question of existence and fast computation of fair and efficient allocations of indivisible resources among agents with additive valuations. As such allocations may not exist for arbitrary instances, we ask if they exist for \textit{typical} or \textit{random} instances, meaning when the utility values of agents for the resources are drawn from certain distributions. If such allocations exist with high probability for typical instances, and furthermore if they can be computed efficiently, this would imply that we could quickly resolve a real world resource allocation scenario in a fair and efficient manner with high probability. This implication has made this setting popular and well studied in fair resource allocation.

In this paper, we extend the previously studied formal models of this problem to non-identical items. We assume that every item is associated with a distribution $\mathcal{U}_j$, and every agent's utility value for the item is drawn independently from $\mathcal{U}_j$. We show that envy-free fair and maximum social welfare efficient allocations exist with high probability in the asymptotic setting, meaning when the number of agents $n$ and items $m$ are large. Further we show that when $m=O(n\log n),$ then by only sampling $O(\log m)$ or $O((\log m)^2)$ utility values per item instead of all the $n,$ we can compute these allocations in $\Tilde{O}(m)$ time. Finally, we simulate our algorithms on randomly generated instances and show that even for small instances, we suffer small multiplicative losses in the fairness and efficiency guarantees even for small sized instances, and converge to fully optimal guarantees quickly.

\end{abstract}


% START -----------------------------------------------------------

\input{1.intro}

\input{2.tech-overview}
\input{3.preliminaries}
\input{4.envy-free-mn}
\input{5.envy-free-small}
\input{6.envy-free-sampling}
\input{7.envy-free-large-main}
\input{8.prop-small}
\input{11.empirical-results}
\input{9.discussion-future}

% END -------------------------------------------------------------

% (if acknowledgements are required) START ------------------------
%\section*{Acknowledgments}
% END -------------------------------------------------------------
\bibliography{aaai2026}  

\input{10.appendix}

\end{document}

%% file: 1.intro.tex
\section{Introduction}
Fair and efficient resource allocation is a ubiquitous problem, for example, when allocating desirable items like assets among companies or public houses among residents, or disposing undesirable items like rent among roommates or online requests among servers. Due to its widespread prevalence, it has been extensively studied in algorithmic game theory, economics and computational social choice (see \mbox{\cite{Amanatidis23survey}} for a survey of the most recent problems and directions).

In this paper we will focus on arguably the strongest known mathematical properties to measure fairness and efficiency, namely, \textit{envy-freeness} and \textit{maximum social welfare}. An allocation is called envy-free if every agent prefers their own share more than any other agent's share. While the social welfare of an allocation is the sum of values of all the agents for their own shares; maximizing the social welfare essentially results in assigning items to the agents who value them the most. An allocation that is both envy-free and has the maximum social welfare would be highly desirable in any allocation scenario. We therefore study the natural question,

\textit{Do envy-free and maximum social welfare allocations exist and can they be efficiently computed?}

There are simple examples to see that such solutions may not always exist. For instance, suppose we have $2$ items and two agents $\{a_1,a_2\}.$ Suppose $a_1$ has value $1$ for both the items and $a_2$ has value $10$ for both. Then assigning both the items to $a_2$ is the only maximum social welfare allocation, and assigning $1$ item to each agent are the only envy-free allocations. In the light of this impossibility, research has been done by relaxing the problem in some way. Some of the resulting, not necessarily mutually exclusive popular directions are, (i) study more tractable notions for fairness or efficiency, like envy-freeness up to one item (proposed by \mbox{\cite{Budish11}}), or (ii) consider a random rather than arbitrary instance (initiated in \mbox{\cite{Dickerson_Goldman_Karp_Procaccia_Sandholm_2014})}, or (iii) characterize subclasses of instances where such solutions always exist, for example by limiting agent utilities (e.g., only allow $0$ or $1$ utilities) and/or the number of agents (e.g., only two agents). 

In this paper, we focus on studying \textit{typical} or \textit{random} instances \mbox{\cite{Dickerson_Goldman_Karp_Procaccia_Sandholm_2014}}, where agent utilities for the resources are drawn from certain distributions rather than being arbitrary. We assume that we are given oracle access to a distribution $\mathcal{U}_j$ (that has minimal technical assumptions described in Section \mbox{\ref{sec:intro-model}}) corresponding to every item $j$. Every agent’s value for item $j$ is sampled independently from its distribution $\mathcal{U}_j$. This model reasonably captures most real world allocation scenarios. For example, while distributing assets among companies, every company's value for a particular asset would be similar, and can be assumed to belong to some natural distribution with bounded upper and lower values. The same is true for resources like public housing, rent and time to serve online server requests. If we can guarantee, with high probability, the existence of envy-free and maximum welfare allocations for random instances drawn in this way, and further give efficient algorithms to find these, then we would essentially be able to quickly resolve any random real world fair allocation problem with high probability. This implication is the strong reason to pursue this direction, thus the motivation for this paper as well.

This setting has been well-studied in several novel works, beginning with \mbox{\cite{Dickerson_Goldman_Karp_Procaccia_Sandholm_2014}}. They initiated this direction with the setting of identical goods, where the distributions for the utilities of all the items were identical. Their work was followed up by \mbox{\cite{Manurangsi_2020}} and \mbox{\cite{manurangsi2020closing}} for the same setting and closely related distribution classes. \mbox{\cite{manurangsi2025asymptoticfairdivisionchores}} then studied the problem with identical undesirable resources, termed chores, instead of goods, and showed stronger positive results, proving that the chores case was easier than goods. \mbox{\cite{Bai_2022}} extended this research to the asymmetric agents case, where each agent was associated with a distinct distribution and drew utilities for all the items independently from their own distribution. They also additionally explore Pareto optimality, a notion of economic efficiency. Finally, \mbox{\cite{Manurangsi2025weighted}} explore another general setting with weighted agents where each agent has a scalar value called a weight assigned to them, and intuitively, fair allocations must assign shares proportional to agent weights. In terms of efficient computation, \mbox{\cite{manurangsi2020closing}} show that an envy-free allocation can be obtained with a round robin algorithm for $m= \Omega(n\log n/\log \log n).$ \mbox{\cite{li2024roundrobin}} perform a detailed analysis of the round robin algorithm, and show that it has complexity $O(mn + m\log m)$ when agent preferences are uniform, identical and random.

We extend this theory in two directions. First we initiate the study of this model to the setting where the items, that is their utility distributions, are allowed to be non-identical. And second, we explore the problem of fine-grained efficient computation of solutions, that is, algorithms that have faster than $O(mn)$ or higher run time. We show the existence of fair and efficient allocations asymptotically, meaning when $m$ and $n$ are large. Furthermore, when $m$ is $O(n\log n),$ we show that sampling only $O(\log m)$ utility values per item suffices to compute these allocations. This reduction in time from $O(mn)$ to $\Tilde{O}(m)$ is highly relevant in large allocation scenarios like server requests or online marketplaces. We also simulate our algorithms by randomly generating instances for various classes of distributions, and show that our theoretical asymptotic guarantees are achieved in practice as well for most instances. We now formally describe our model and contribution in detail.

\subsection{Our Model}\label{sec:intro-model}

\noindent
\paragraph{Fair allocation problem.} We consider the problem of dividing $m$ indivisible and non-identical items among $n$ agents, according to specified fairness and efficiency notions (defined shortly). We assume the items are \textit{goods,} meaning they are non-negatively valued by all the agents. Formally, let $N = \{1, \dots, n\}$ be the set of agents and $M = \{1, \dots, m\}$ be the set of goods. For each $(i, j) \in (N \times M)$, $u_i(j) \in [0, 1]$ denotes the \textit{utility} of good $j$ for agent $i$. The utilities are assumed to be \textit{additive}, that is, for all $M' \subseteq M$, $u_i(M') = \sum_{j \in M'} u_i(j)$. The solution of the fair allocation problem is an allocation or a partition of $M$ into $n$ disjoint subsets. Formally, we denote an allocation by $A = (A_1, \dots, A_n),$ where $A_i$ is the bundle or subset of $M$ allocated to agent $i.$ Thus $\cup_{i} A_i = M$ and $A_i \cap A_{i'} = \phi$ for all $i$ and $i' \in N$. In this paper, we will study the setting where each utility $u_i(j)$ is sampled randomly from a known distribution $\mathcal{U}_j,$ described below. 

\noindent
\paragraph{Distributions of utilities.} For each good $j \in M$, we assume that there exists a distribution $\mathcal{U}_j$ supported on $[0,1]$, such that $u_i(j) \sim \mathcal{U}_j$ independently for each $i \in N$. Computationally, we assume a black box for each $\mathcal{U}_j$ that allows us to sample once from the distribution in $O(1)$ time. The distributions have the following properties. Each $\mathcal{U}_j$ is \textit{non-atomic}, i.e., $\Pr_{X \sim \mathcal{U}_j}[X = x] = 0$ for all $x \in [0, 1]$. Further, every $\mathcal{U}_j$ is $(\alpha_j, \beta_j)$-PDF bounded, meaning the density function of each $\mathcal{U}_j$ has value between $\alpha_j$ and $\beta_j$ at every point in its support, for some constants $\alpha_j, \beta_j>0$. Here the $\alpha_j$ and $\beta_j$ may be distinct for each $j \in M$.

% Let $m = xn + y$ for some $x \in \mathbb{Z}_{\geq 0}$ and $0 \leq y < n$.
\noindent
\paragraph{Fairness and efficiency notions.} We study the existence and computability of the following three types of allocations. An allocation $A$ is \textit{envy-free} if for all pairs of agents $i, i' \in N$, it holds that $u_i(A_i) \geq u_i(A_{i'})$. The envy of any agent $i$ towards another agent $i'$ is $\max \{0, u_i(A_{i'}) - u_i(A_i)\}$. An allocation $A$ is \textit{proportional} if for all agents $i \in N$, it holds that $u_i(A_i) \geq \nicefrac{u_i(M)}{n}$. Finally, an allocation $A$ has \textit{maximum social welfare} if it maximizes the sum of agent values for their own shares. Formally, if $\mathcal{A}$ is the set of all allocations of $M,$ then $A\in argmax_{\mathcal{A}} \sum_i u_i(A_i)$. We also consider equivalent approximation notions for these. An allocation $A$ is \textit{c-approximate envy-free} ($c < 1$) if for all pairs of agents $i, i' \in N$, it holds that $u_i(A_i) \geq c \cdot u_i(A_{i'})$. Finally, an allocation $A$ has \textit{c-approximate maximum social welfare} ($c < 1$) if some allocation $A'$ has maximum social welfare and $\sum_i u_i(A_i) \geq c\ \cdot \sum_i u_i(A'_i)$.

%% file: 2.tech-overview.tex
\section{Our Contribution and Techniques}\label{sec:intro-contribution}
Our results can be separated into two parts, based on whether the number of items is large or small compared to the number of agents. The large items case is solved by (i) sampling a few agents and their utilities for each item and (ii) assigning it to the agent who values it the most among these samples. The small items case is significantly more challenging, and solved by reducing the problem to that of finding what are called perfect $r$-matchings and right-saturated $r$-matchings, both being generalizations of perfect matchings in bipartite graphs. %We will now overview all our results and proofs.  

\subsection{Large Number of items}
We show four results for the case when $n = O(m/\log m),$ the first of which proves the existence of fair and efficient allocations. Then three results, each with slightly more specific but reasonable assumptions on the distribution model,  show that sampling $O(\log m)$ or $O((\log m)^2)$ utility values yields near optimal guarantees in much faster $\tilde{O}(m)$ time.

\begin{restatable}{theorem}{eflargegoods}
\label{thm:Goods-Large-m}
When the number of goods is large, that is, $n = O(m/\log m),$ then, with probability $(1-\nicefrac{1}{m})$, an envy-free and maximum social welfare allocation exists as $m\rightarrow \infty.$
\end{restatable}

To prove this, we show that when all utility information is available (every agent evaluates every good) then simply assigning each good to the agent who values it the most yields an allocation that has the maximum social welfare by construction and is also envy-free. The proof builds on a gap in expectations on utilities conditioned on whether they are the highest out of $n$ i.i.d. draws or not - for i.i.d. draws from a continuous, non-atomic utility distribution, the expected utility of an agent who wins a maximum is strictly higher than that of an agent who loses. This utility gap is exploited to show that, in expectation, every agent values their own bundle more than any other agent's. To move from expectation to a high probability guarantee for any pair of agents $i, i'$ ($\Pr \left[ i \text{ envies } i' \right] \leq \nicefrac{1}{m^3}$), we use concentration bounds (Lemma \ref{lem:Chernoff}), both for an agent $i$'s valuation of their bundle $A_i$ as well as the bundle $A_{i'}$ of agent $i'$. A union bound over all pairs of agents completes the argument and gives us the desired result ($\Pr \left[ \text{envy free allocation} \right] \geq 1 -\nicefrac{1}{m}$). Notably, this result makes no assumptions on the specific form of utility distributions beyond non-atomicity and independence.

The proof is almost the same as for the case when the distributions are identical \cite{Dickerson_Goldman_Karp_Procaccia_Sandholm_2014}. The key difference is in aggregating the expected utility of an agent's bundle from their expected utility from a single good, as there are non-identical goods now. This aggregation has to be done for the agent's value for their own bundle, as well as for another agent's bundle, for analyzing envy. 

\subsection{Sampling Algorithms and Experiments}

As the previous result holds when $m\rightarrow \infty,$ thus for large inputs, it is natural to study efficient computation in a more fine-grained manner, and obtain asymptotically more efficient algorithms. Note that there is a natural lower bound of $O(m)$ for both time and space, as every item needs to be processed once to decide which agent it gets assigned to. We explore if we can match this lower bound. Towards this, we assume an online model for the large number of goods case, where the items appear one by one, and must be allocated immediately and irrevocably upon arrival. For this setting, we consider the question, \textit{`Can we obtain envy-free and maximum social welfare allocations by sampling a few utility values, and in almost O(m) time?'} We give the following positive answers. %for a related model which has discrete instead of continuous distributions for each item, and an almost positive answer for the continuous distributions case, where we can obtain near optimal solutions. Our algorithms run in $\tilde{O}(m)$ time, thus match the lower bounds up to logarithmic factors. Our results are as follows.
\begin{restatable}{theorem}{samplingdiscrete}
\label{thm:sampling-discrete}
    Suppose the utility of each good for every agent is drawn from a discrete distribution with finite support in $[0,1]$. Let \( m = \Omega \left(n \log{n} \right) \) and the number of sampled agents for each good \( s = \frac{2 \log{m}}{\alpha_{min}} \). Then, with high probability, a maximum social welfare and envy-free online allocation of goods exists as \( n \to \infty \).
\end{restatable}
\begin{restatable}{theorem}{samplingcontinuousconstant}
\label{thm:sampling-continuous-constant}
For the allocation problem as described in Section \ref{sec:intro-model}, let \( m = \Omega \left(n \log{n} \right) \) and the number of sampled agents for each good \( s = \frac{20 \log{m}}{\alpha_{min}} \). Then, with high probability, a $0.9$-approx. maximum social welfare and $0.8$-approx. envy-free online allocation of goods exists as \( n \to \infty \).
\end{restatable}
\begin{restatable}{theorem}{samplingcontinuous}
\label{thm:sampling-continuous}
Let \( m = \Omega (n \log{n}) \) and the number of sampled agents for each good \( s = \frac{2 ( \log{m})^2}{\alpha_{min}} \). If, for all goods $j$, $\mathbb{E}[ \mathcal{U}_j] \leq \mu$ for some constant $\mu < 1$, then, with high probability, a $( 1 - \frac{1}{\log {n}})$-approx. maximum social welfare and envy-free online allocation of goods exists as \( n \to \infty \).
\end{restatable}
In the first model, we assume that each utility distribution is discrete, supported on finitely many values in $\left[ 0, 1 \right]$. The allocation mechanism samples $\Theta (\log{m})$ agents for each good and assigns the good to the agent with the highest utility among those sampled. Using $(\alpha, \beta)$ boundedness, we ensure that with high probability, each good is assigned to an agent who values it at exactly $1$. In the second model, now as the distributions are continuous, we can only guarantee a value close to $1$ with high probability and not exactly $1$ from each item. Thus we obtain approximate envy-freeness and approximate maximum social welfare. Specifically, we ensure that each good is most likely assigned to an agent who values it above a fixed threshold of $0.9$, thereby yielding approximate maximum social welfare. In both the results, the analysis of envy proceeds similarly as the offline case. The only straightforward modifications are that as we do not have full knowledge of all utilities, we now quantify the utility gap between agents using the additional assumptions on the distributions.

In the final result, we observe the following additional property that our model implies: the continuous utility distributions have an upper bound on their mean (strictly less than $1$). We prove that, as long as the means are bound by a fixed constant, a sampling-based allocation still yields fairness and efficiency up to vanishing error ($1 - \frac{1}{\log{n}}$), provided the sample size is poly-logarithmic. The main technical idea to prove this theorem is to relate the maximum utility observed in a sample to the upper bound on the mean. With a sufficiently large sample, the maximum observed utility exceeds the mean with high probability. This creates a utility gap between the winner's expected utility and that of any other agent. By controlling this gap and using concentration bounds as in the previous proofs, we ensure that envy is unlikely to occur even in this general setting. This result highlights the robustness of sampling: even when the distributions are unknown or highly skewed, approximate fairness and efficiency can be achieved as the number of goods grows.

\noindent
\textit{Empirical Results.} We empirically evaluate both the non-sampling and sampling algorithms on finite instances to understand how quickly the asymptotic behavior manifests. The details of our set up are in Section \mbox{\ref{sec:empirical-results}}. We could conclude that even in practice, large item counts rapidly eliminate envy, and sampling preserves both fairness and efficiency.

%These results demonstrate that simple, greedy allocation mechanisms, both in full-information and limited-information settings, can guarantee envy-freeness and efficiency in the large market limit. Importantly, the guarantees hold not just in expectation, but \textit{ex-post} with high probability, making them robust and practically applicable.

%For each item, we show that instead of looking at every agent's utility for the item, it suffices to pick $O(\log m)$ agents per item uniformly at random, and assign the item to the agent who has the highest value for the item among this sampled set (See Theorem \ref{thm:sampling-continuous-constant}). This yields a near optimally $0.8$-envy-free and $0.9$-efficient allocation. Further, if we allow  $O((\log m)^2)$ samples per item, we get a fully envy-free and $(1-1/\log m)$-efficient allocation (See Theorem \ref{thm:sampling-continuous}). Finally, if we consider a slightly different model, with discrete distributions instead of continuous, then sampling $O(\log m)$ samples per item yields an envy-free and maximum social welfare allocation (Theorem \ref{thm:sampling-discrete}). So long as the number of goods is sub-exponentially larger than the number of agents, particularly for any polynomially larger $m = poly(n),$ these algorithms are $\tilde{O}(m),$ thus match the lower bound up to logarithmic factors. 

\subsection{Small Number of Items}
For the complementary case for smaller $m$, we show that results similar to those obtained with identical items generalize to the non-identical case as well. For the case of fairly allocating identical \textit{goods}, \mbox{\cite{Manurangsi_2020}} show that an envy-free allocation of \textit{goods} only exists when $m = \Omega(n \log n/\log \log n),$ and only exists for smaller $m$ when $m \ge 2n$ and is divisible by $n$.

\begin{restatable}{theorem}{efsmallgoods} 
\label{thm:Goods-Small-m}
When $m \ge 5n$ and $m$ is divisible by $n,$ then with high probability, an envy-free allocation exists as $n\rightarrow \infty$.
\end{restatable}

The key intuition to prove this theorem is that a balanced allocation - where each agent receives an equal number of goods - is highly likely to be envy-free when the allocation ensures that every agent gets only those goods that they value extremely close to $1$ ($\geq 1 - \frac{1.1 \log{n}}{\alpha_{min} n} = \tau$). We show the existence of such an allocation by showing that structures called $r$-matchings in some random bipartite graphs correspond to such a balanced allocation, and use a result about such matchings \cite{erdHos1964random}. We create a bipartite graph with an edge between any agent $i$ and good $j$ if and only if $u_i(j) \geq \tau$. The probability of such an edge existing is at least $\nicefrac{(\log{n} + \omega(1))}{n},$ exactly what is required to prove the existence of a perfect $r$-matching. The proof builds upon similar ideas for the case of identical distributions from \cite{Manurangsi_2020}. However, a key difference is that instead of removing certain edges from the constructed bipartite graph to ensure envy-freeness, we simply show that with high probability, for any pair of agents $i, i'$, strictly less than half the goods in $i'$s bundle are valued by $i$ above a carefully chosen threshold, and thus envy-freeness still holds.

%Before presenting our other results, we would like to point out three related works in existing literature and comment how this result compares with these. First,  \cite{Manurangsi_2020} have shown a stronger non-existence result for all $m\le O(n\log n/\log \log n).$ Their model assumes a weaker condition than $(\alpha, \beta)$-boundedness, hence for this larger class of distributions they have a stronger negative result. Their model however also assumes identical items, hence is not directly comparable to ours. Another related work \cite{manurangsi2020closing} proves a positive result that envy-free allocations exist when $m = O(n\log n/\log \log n).$ This paper models the same distribution as ours and also assumes identical items. We improve this positive result in two ways, by showing existence even when $m\ge 5n,$ and with non-identical items.

Finally, we obtain the next best fairness guarantee in the remaining cases, (i) when $m < 5n$ or when (ii) $m \ge 5n$ and $m$ is not divisible by $n.$ Non-existence of envy-free allocations in the latter case has been shown by \cite{manurangsi2020closing}. We show that proportional allocations always exist for most $m$ in these ranges.

\begin{restatable}{theorem}{propsmallgoods} 
\label{thm:prop-small-m}
For goods, if $n \le m \le 2n$ and for all $j,$ $\mathcal{U}_j$ has mean at most $\nicefrac{1}{2},$ then, with high probability, a proportional allocation always exists as $n\rightarrow \infty.$
\end{restatable}

This is the most technically involved result of this paper. First, by a simple pigeonhole principle argument, we prove the sufficiency of the condition on the means of the distributions. For the case when this condition is satisfied, we design a two-stage matching algorithm to constructively prove a proportional allocation exists. In the first stage, the algorithm constructs a bipartite graph that connects every agent $i$ with any good $j$ only if $u_i(j) \geq \tau_j$, where $\tau_j$ is carefully chosen so that $\Pr\bigl[u_i(j) \geq \tau_j\bigr]$ is exactly $\nicefrac{\log{n} + \omega(1)}{n}$. Due to Lemma~\ref{lem:Erdos-Renyi}, this graph has a perfect matching. Now, for the case where $m < 1.9n$, we show by simple concentration bounds that this matching is already proportional.

When $m \geq 1.9n$, we need a tighter analysis. Using the Berry-Esseen theorem for independent but not necessarily identically distributed summands, we show that at most $0.6n$ agents remain without a proportional share after the first stage. For these agents, we construct a residual graph with the remaining agents and items, and via a similar analysis as in the previous case, we show that this graph helps to complete a proportional allocation.

This proof is inspired from a similar analysis by \cite{manurangsi2020closing} for a single distribution case. Significant work is required to generalize the analysis to the non-identical distributions case. Along the way, we proved a stronger concentration bound using the Berry-Esseen inequality to show that the unsatisfied agents in the second case are less than $0.6n.$ This was necessary as the utilities are from non-identical distributions, and bounding the probability that a utility sample falls outside the proportionality bar is quite involved.

\begin{restatable}{theorem}{proplargegoods}
\label{thm:prop-large-m}
Let $c < 1$ be such that $\mathbb{E}[\mathcal{U}_j] \leq c$ for all $j \in M$. If $r = \lceil 2 \cdot \frac{3 + c}{1 - c} \rceil$ and $m \geq rn$, then with high probability, a proportional allocation exists as $n\rightarrow \infty$.
\end{restatable}

The idea for this result is similar to the case of envy-freeness for small number of items. We first obtain a balanced allocation - each agent has $r$ goods - where each good has value close to $1$ for the agent. We show that this allocation itself is proportional with high probability. To do this, we show that every agent values at most half of all the goods greater than a chosen threshold of $(1+c)/2$, thereby showing an upper bound on any agent's proportional value.

\noindent
\paragraph{Organization.} In the remaining paper, we discuss our main results for each case, large and small number of goods - Theorem \mbox{\ref{thm:prop-small-m}} for a proportional allocation when there are $n\leq m\leq 2n$ goods in Section \mbox{\ref{sec:prop-small}}, Theorem \mbox{\ref{thm:Goods-Small-m}} for envy-free allocations with small number of goods in Section \mbox{\ref{sec:ef-small-goods}}, and Theorem \mbox{\ref{thm:sampling-continuous}}, our sampling algorithm for large number of items with continuous distributions and bounded means in Section \mbox{\ref{sec:ef-sample}}. We also describe our empirical evaluations in Section \mbox{\ref{sec:empirical-results}}. Our proofs describe all the novel ideas in the paper and essentially subsume the remaining proofs. For completeness, we present them in the Appendix.

%% file: 3.preliminaries.tex
\section{Preliminaries}

In the analysis of all of our results, we will use the following fundamental probability bounds.
\begin{lemma} [Chernoff] \label{lem:Chernoff}
    Let $X_1, \dots, X_m$ be independent random variables in $\left[0 , 1 \right]$. Denote $X = \sum_{i=1}^{m} X_i$. Then for all $\epsilon \in [0, 1]$, 
    $ \Pr [ X \leq (1 - \epsilon) \mathbb{E} [X]] \leq \exp{( - \frac{\epsilon^2}{2} \mathbb{E}[X])},$ and 
    $ \Pr [ X \geq (1 + \epsilon) \mathbb{E} [X]] \leq \exp{( - \frac{\epsilon^2}{3} \mathbb{E}[X])}.$
\end{lemma}

\begin{lemma} \label{lem:Hoeffding's inequality}
    \cite{hoeffding1963probability}. 
    Let $X_1, \dots, X_n$ be independent random variables such that $X_i \in [a_i,b_i]$. Let $S_n = \sum_{i=1}^n X_i$ and $\mathbb{E}[S_n] = \mu$. Then for all $t>0,$ $\Pr[S_n - \mu \geq t] \leq \exp{\left(-\frac{2t^2}{\sum_{i=1}^n (b_i - a_i)^2}\right)}.$
\end{lemma}

The proof of Theorem \ref{thm:prop-small-m} also uses two stronger probabilistic inequalities, which we state here.

\begin{lemma}[Lyapunov's Moment Inequality, \cite{karrNotes} Corollary 4.110, p.182]\label{lem:Lyapunov's inequality}
    Let \( X_j \) be a real-valued random variable with \( \mathbb{E}[X_j^2] < \infty \). Then,
    \[
    \mathbb{E}[|X_j|^3] \leq \left( \mathbb{E}[X_j^2] \right)^{3/2} = \sqrt{\mathbb{E}[X_j^2]} \cdot \mathbb{E}[X_j^2].
    \]
\end{lemma}

\begin{lemma}[Berry-Esseen Inequality, ~\cite{esseen1942liapunoff}]\label{lem:Berry-Esseen}
    Let \( X_1, \ldots, X_n \) be independent real-valued random variables with zero mean and finite third absolute moments. Let
    $
    S_n = \sum_{j=1}^n X_j
    $
    and
    $
    \sigma_n^2 = \sum_{j=1}^n \mathbb{E}[X_j^2] > 0.
    $
    Let \( F_n \) denote the distribution function of \( S_n / \sigma_n \), and let \( \Phi \) denote the standard normal distribution function. Then,
    \[
    \sup_x \left| F_n(x) - \Phi(x) \right| \leq C_0 \cdot \frac{ \sum_{j=1}^n \mathbb{E}[|X_j|^3] }{ \left( \sum_{j=1}^n \mathbb{E}[X_j^2] \right)^{3/2} }.
    \]
    where $C_0$ is fixed constant.
\end{lemma}

\subsection{Matchings in Random Bipartite Graphs}

We now discuss the random graphs and their properties used in analyzing the case of small number of items. 

Let $\mathcal{G} (n_L, n_R, p)$ denote the distribution of random bipartite \\ graphs with $n_L$ left vertices, $n_R$ right vertices, and $p$ being the probability that there exists an edge between any pair of left and right vertices. A graph sampled from this distribution is called an \textit{Erd{\H{o}}s-R\'enyi random bipartite graph}. The following result is classical and well-known.
\begin{lemma}\label{lem:Erdos-Renyi}
    \cite{erdHos1964random}. Let $G \sim \mathcal{G} (n, n, p)$ such that $p = \nicefrac{(\log{n} + \omega(1))}{n}$. Then, with high probability, $G$ contains a perfect matching.
\end{lemma}

We prove an extension to Lemma \ref{lem:Erdos-Renyi} where the probability of edges depends on the right vertices.
\begin{theorem}\label{thm:Extension to Erdos Renyi}
    Let $G \sim \mathcal{G} (n, n, \{p_j\}_{j \in [n]})$ be a random bipartite graph, say $G = (\{l_i\}_{i \in [n]}, \{r_j\}_{j \in [n]}, E)$. Each edge $(l_i, r_j) \in E$ independently with probability $p_j$. Define $p_{\min} := \min_{j \in [n]} p_j$. If $p_{\min} \geq \nicefrac{(\log{n} + \omega(1))}{n},$ then with high probability $G$ contains a perfect matching.
\end{theorem}
\begin{proof}
    Given $G$, define a new random bipartite graph $G' = (\{l_i\}_{i \in [n]}, \{r_i\}_{i \in [n]}, E')$ as follows. If $(l_i, r_j) \in E$, then keep the edge $(l_i, r_j)$ in $G'$ with probability $\nicefrac{p_{min}}{p_j}$, and remove it with probability $1 - \nicefrac{p_{min}}{p_j}$. Observe that for each $(l_i, r_j)$, $\Pr[(l_i, r_j) \in E']=$ $\Pr[(l_i, r_j) \in E] \Pr[(l_i, r_j) \in E' \mid (l_i, r_j) \in E]$ $=p_j \cdot \nicefrac{p_{min}}{p_j} = p_{min}.$
    
    Thus, $G'$ is a random bipartite graph where each edge is included independently with constant probability $p_{\min}$. Since $p_{\min} \geq \nicefrac{(\log{n} + \omega(1))}{n}$, from Lemma \ref{lem:Erdos-Renyi}, $G'$ contains a perfect matching with high probability. Further, $E' \subseteq E$, so $G$ also has a perfect matching with high probability.
\end{proof}

Next, we state the following known result, which proves some sufficient conditions that random graphs must have to satisfy Hall's condition for perfect matchings.
\begin{lemma}\label{lem:Hall's-condition-satisfying-case}
    \cite{manurangsi2020closing}. Let $G = (A,B,E)$ be a graph sampled from the Erdős-Rényi random bipartite graph distribution $\mathcal{G}(n,q,p)$, where $p \geq 0.5$ and $0.9n \leq q \leq n$. Then with high probability, for every $S \subseteq A$ of size at most $0.6n$, we have $|Z_G(S)| \geq |S|$, where $Z_G(S)$ is the set of vertices adjacent to $S$ in graph $G$.
\end{lemma}

\noindent 
Finally, we talk about right saturated $r$-matchings. 
\begin{definition}
    An $r$-matching of a bipartite graph $G = (L, R, E)$ is a subset of edges $E'\subseteq E$ such that the subgraph $G' = (L, R, E')$ has degree at most $r$ for vertices in $L$ and degree at most $1$ for vertices in $R$.
\end{definition}
\begin{definition}
    A \textbf{right-saturated} $r$-matching of a bipartite graph $G = (L, R, E)$ is a subset of edges $E'\subseteq E$ such that the subgraph $G' = (L, R, E')$ has degree at most $r$ for vertices in $L$ and degree \textbf{exactly} $1$ for vertices in $R$. It is said to be \textbf{perfect} if the subgraph $G'$ has degree \textbf{exactly} $r$ for vertices in $L$ and degree \textbf{exactly} $1$ for vertices in $R$
\end{definition}

\begin{lemma}\label{lem:Perfect-r-Matching}
    Let $G \sim \mathcal{G} (n, m, p)$ such that $p = \nicefrac{(\log{n} + \omega(1))}{n}$ and $m = rn$. Then, with high probability, $G$ contains a perfect $r$-matching.
\end{lemma}
\begin{proof}
    Given $G$, define a new random bipartite graph $G'$ by making $r$ copies of each left vertex as well as the edges incident on the vertices. Clearly $G' \sim \mathcal{G} (m, m, p)$, and by Lemma \ref{lem:Erdos-Renyi}, it contains a perfect matching ($p = \frac{\log{n} + \omega(1)}{n} \geq \frac{\log{m} + \omega(1)}{m}$). This perfect matching in $G'$ corresponds to a perfect $r$-matching in $G$.
\end{proof}

A similar reduction, by adding $n-m$ dummy right vertices and edges to these independently with probability $p,$ helps to prove the following. 
\begin{lemma}\label{lem:Right-Saturated-Matching}
    Let $G \sim \mathcal{G} (n, m, p)$ such that $p = \nicefrac{(\log{n} + \omega(1))}{n}$ and $n > m$. Then, with high probability, $G$ contains a right saturated matching.
\end{lemma}

%% file: 4.envy-free-mn.tex
\section{Envy-freeness for Large number of Goods}\label{sec:ef-large}
Before we prove our envy-freeness result, we need the following result on the expected conditional utilities of goods.

\begin{lemma}\label{lem:Utility-Gap}
    Let $\mathcal{U}$ be a \textbf{non-atomic}, \textbf{continuous} distribution and $u_1, \dots, u_n$ are $n$ i.i.d. draws from $\mathcal{U}$, then for any $i, i' \in [n], i \neq i'$

    \begin{equation*}
        \mathbb{E}[u_i \mid \arg \max_{j \in [n]} u_j = i] > \mathbb{E}[u_i \mid \arg \max_{j \in [n]} u_j = i']
    \end{equation*}
\end{lemma}

\begin{proof}
    For any $i \in [n]$ define the event
    \[
    E_i = \Bigl\{\arg\max_{j \in [n]} u_j = i\Bigr\}
    \]

    Since $\mathcal{U}$ is continuous, there are no ties among $u_1, \dots, u_n$, so exactly one of $E_1, \dots, E_n$ occurs and \(\Pr(E_i) = \frac{1}{n}, \forall i \in [n]\).

    Since $E_1, \dots, E_n$ are mutually exclusive and exhaustive events, we have by the law of total probability,
    \[
    \mathbb{E}[u_i] = P[E_1]\,\mathbb{E}[u_i \mid E_1] + \dots + P[E_n]\,\mathbb{E}[u_i \mid E_n]
    \]
    i.e.,
    \[
    \mathbb{E}[u_i] = \frac{1}{n}(\mathbb{E}[u_i \mid E_1] + \dots + \mathbb{E}[u_i \mid E_n])
    \]

    Now again, since there are no ties among $u_1, \dots, u_n$, we can say that for any $i', i'' \in [n]$ such that $i'\neq i$ and $i''\neq i$, $\mathbb{E}[u_i \mid E_{i'}] = \mathbb{E}[u_i \mid E_{i''}]$. Thus,
    \begin{equation}\label{eq:Ei,i'}
        \mathbb{E}[u_i] = \frac{n-1}{n}\mathbb{E}[u_i \mid E_{i'}] + \frac{1}{n}\mathbb{E}[u_i\mid E_i]
    \end{equation}
    where $i' \neq i$.
    
    On the event $E_i$, we have
    \[
    u_i = \max\{u_1, \dots, u_n\}.
    \]
    
    Denote $M_n = \max\{u_1, \dots, u_n\}$.  Then
    \[
    \mathbb{E}\bigl[u_i \mid E_i\bigr]
    =
    \mathbb{E}\bigl[M_n\bigr].
    \]
    
    Because $\mathcal{U}$ is non‐degenerate (i.e.\ $\mathrm{Var}(\mathcal{U}) > 0$, $\mathcal{U}$ is not a point-mass), it is standard that
    \[
    \mathbb{E}[\,M_n\,] \;>\; \mathbb{E}[\,u_i\,].
    \]
    
    Hence
    \begin{equation}\label{eq:Ei}
    \mathbb{E}\bigl[u_i \mid E_i\bigr] 
    =
    \mathbb{E}[\,M_n\,] 
    > 
    \mathbb{E}[\,u_i\,].
    \end{equation}
    
    Combining \eqref{eq:Ei,i'} and \eqref{eq:Ei} gives
    \[
    \mathbb{E}[u_i \mid E_i]
    >
    \frac{n-1}{n}\mathbb{E}[u_i \mid E_{i'}] + \frac{1}{n}\mathbb{E}[u_i\mid E_i]
    \]
    i.e.,
    \[
    \mathbb{E}[u_i \mid E_i]
    >
    \mathbb{E}[u_i \mid E_{i'}]
    \]
    which is precisely
    \[
    \mathbb{E}[u_i \mid \arg \max_{j \in [n]} u_j = i] 
    > 
    \mathbb{E}[u_i \mid \arg \max_{j \in [n]} u_j = i']
    \]
\end{proof}

\eflargegoods*

% \begin{theorem}\label{thm:Goods-Large-m}
%     Let \( m = \Omega \left(n \log{n} \right) \). Then a Pareto optimal and EF (envy-free) allocation of goods exists with high probability as \( n \to \infty \).
% \end{theorem}

\begin{proof}
    Construct an allocation by assigning each good \( j \in M \) to the agent who values it most: \( \arg \max_{k \in N} u_k(j) \). The allocation has maximum social welfare by construction. We now show that this allocation is EF with high probability.

    For every good \( j \in M \), each agent \( i \in N \) draws their utility \( u_i(j) \sim \mathcal{U}_j \), where \( \mathcal{U}_j \) is a \textit{non-atomic}, \textit{continuous} distribution with support in \( [0, 1] \). By Lemma \ref{lem:Utility-Gap} for each \( j \in M \), there exist constants \( \mu_j \) and \( \mu_j^* \) such that, for any \( i, i' \in N \):

    \begin{align*}
    & 0 < \mathbb{E} \left[u_i(j) \, \middle| \, \arg \max_{k \in N} u_k(j) = \{i'\} \right] \leq \mu_j \\
    & < \mu_j^* \leq \mathbb{E} \left[u_i(j) \, \middle| \, \arg \max_{k \in N} u_k(j) = \{i\} \right]
    \end{align*}
    
    Define for fixed \( i \in N \), \( j \in M \):
    
    \[
    X_i^j =
    \begin{cases}
    u_i(j), & \text{if } \arg \max_{k \in N} u_k(j) = \{i\} \\
    0, & \text{otherwise}
    \end{cases}
    \]
    
    Then \( u_i(A_i) = \sum_{j \in M} X_i^j \). Moreover,
    
    \begin{align*}
    \mathbb{E}[X_i^j] & = \Pr\left[\arg \max_{k \in N} u_k(j) = \{i\}\right] \cdot \\ 
    & \mathbb{E}\left[ u_i(j) \,\middle|\, \arg \max_{k \in N} u_k(j) = \{i\} \right] \\
    & = \frac{1}{n} \cdot \mathbb{E}\left[ u_i(j) \,\middle|\, \arg \max_{k \in N} u_k(j) = \{i\} \right] \geq \frac{\mu_j^*}{n}
    \end{align*}
    
    Using linearity of expectation,
    
    \begin{align*}
    & \mathbb{E}[u_i(A_i)] = \sum_{j \in M} \mathbb{E}[X_i^j] \geq \mu_a^* \cdot \frac{m}{n}, \quad \\ 
    & \text{where } \mu_a^* = \frac{1}{m} \sum_{j \in M} \mu_j^*
    \end{align*}
    
    Now, define for \( i, i' \in N \), \( i \neq i' \), and \( j \in M \):
    
    \[
    Y_{ii'}^j =
    \begin{cases}
    u_i(j), & \text{if } \arg \max_{k \in N} u_k(j) = \{i'\} \\
    0, & \text{otherwise}
    \end{cases}
    \]
    
    Then \( u_i(A_{i'}) = \sum_{j \in M} Y_{ii'}^j \). Furthermore,
    
    \begin{align*}
    \mathbb{E}[Y_{ii'}^j] 
    &= \frac{1}{n} \cdot \mathbb{E}\left[ u_i(j) \,\middle|\, \arg \max_{k \in N} u_k(j) = \{i'\} \right] \leq \frac{\mu_j}{n}
    \end{align*}
    
    Let \( Z_{ii'}^j \in [0,1] \) be such that \( \mathbb{E}[Z_{ii'}^j] = \frac{\mu_j}{n} \), and \( Z_{ii'}^j \) stochastically dominates \( Y_{ii'}^j \). Then:
    
    \[
    \sum_{j \in M} \mathbb{E}[Z_{ii'}^j] = \mu_a \cdot \frac{m}{n}, \quad \text{where } \mu_a = \frac{1}{m} \sum_{j \in M} \mu_j
    \]
    
    and for all \( x \in \mathbb{R}^+ \):
    
    \[
    \Pr\left[\sum_{j \in M} Z_{ii'}^j \geq x\right] \geq \Pr\left[\sum_{j \in M} Y_{ii'}^j \geq x\right]
    \]
    
    Let \( E_{ii'} \) be the event that agent \( i \) envies agent \( i' \), i.e., \( \sum_{j \in M} Y_{ii'}^j > \sum_{j \in M} X_i^j \). This occurs only if:
    
    \begin{align*}
    \sum_{j \in M} X_i^j 
    & \leq \mu_a^* \frac{m}{n} - \frac{\mu_a^* - \mu_a}{2} \frac{m}{n} \\ 
    & = \left(1 - \frac{\mu_a^* - \mu_a}{2\mu_a^*}\right) \mu_a^* \frac{m}{n} \\ 
    & \leq \left(1 - \frac{\mu_a^* - \mu_a}{2\mu_a^*}\right) \cdot \mathbb{E}\left[\sum_{j \in M} X_i^j\right]
    \end{align*}
    
    or
    
    \begin{align*}
    \sum_{j \in M} Y_{ii'}^j 
    & \geq \mu_a \frac{m}{n} + \frac{\mu_a^* - \mu_a}{2} \frac{m}{n} \\ 
    & = \left(1 + \frac{\mu_a^* - \mu_a}{2\mu_a}\right) \mu_a \frac{m}{n} \\
    & = \left(1 + \frac{\mu_a^* - \mu_a}{2\mu_a}\right) \cdot \mathbb{E}\left[\sum_{j \in M} Z_{ii'}^j\right]
    \end{align*}
    
    Let \( \epsilon = \min\left\{1, \frac{\mu_a^* - \mu_a}{2\mu_a^*} \right\} \), so \( \epsilon < \frac{\mu_a^* - \mu_a}{2\mu_a} \) also holds. $\epsilon > 0$ also holds as $\mu_j < \mu_j^*$ for all $j$ implies $\mu_a < \mu_a^*$.
    
    Since \( X_i^j \) and \( Z_{ii'}^j \) are independent across \( j \), using Chernoff bounds:
    
    \begin{align*}
    & \Pr\left[ \sum_{j \in M} X_i^j \leq (1 - \epsilon) \cdot \mathbb{E}\left[\sum_{j \in M} X_i^j\right] \right] \\
    & \leq \exp\left(- \frac{\epsilon^2}{2} \mu_a^* \cdot \frac{m}{n} \right)
    \end{align*}
    
    and
    
    \begin{align*}
    & \Pr\left[ \sum_{j \in M} Y_{ii'}^j \geq (1 + \epsilon) \cdot \mathbb{E}\left[\sum_{j \in M} Z_{ii'}^j\right] \right] \\ 
    & \leq \exp\left(- \frac{\epsilon^2}{3} \mu_a \cdot \frac{m}{n} \right)
    \end{align*}
    
    Set \( n \leq \frac{\epsilon^2 \mu_a}{3} \cdot \frac{m}{\ln(2m^3)} \). By the union bound:
    
    \begin{align*}
    \Pr[E_{ii'}] 
    &\leq \exp\left(- \frac{\epsilon^2}{2} \mu_a^* \cdot \frac{m}{n} \right) + \exp\left(- \frac{\epsilon^2}{3} \mu_a \cdot \frac{m}{n} \right) \\
    &\leq 2 \cdot \frac{1}{2m^3} = \frac{1}{m^3}
    \end{align*}
    
    Allocation \( A \) is EF if and only if no event \( E_{ii'} \) occurs. The probability that \( A \) is not EF is at most:
    
    \[
    \Pr\left[\bigvee_{\substack{i, i' \in N \\ i \neq i'}} E_{ii'}\right] \leq \sum_{\substack{i, i' \in N \\ i \neq i'}} \Pr[E_{ii'}] \leq \binom{n}{2} \cdot \frac{1}{m^3} \leq \frac{1}{m}
    \]
    
    Thus, the probability that \( A \) is not EF vanishes as \( m \to \infty \) (\( n \to \infty \) and \(m = \Omega (n \log{n})\)).
\end{proof}

%% file: 5.envy-free-small.tex
\section{Envy-freeness for Small Number of Goods}\label{sec:ef-small-goods}

\efsmallgoods*

% \begin{theorem}
%     If $m \geq 5n$, then with high probability, there exists an envy-free allocation of goods.
% \end{theorem}

\begin{proof}
    Let $m = xn$ for some $x \in \mathbb{Z}_{\geq 0}$. Let $\alpha_{min} = \min_{j \in M} \alpha_j$ and $\beta_{max} = \max_{j \in M} \beta_j$.

    Construct a bipartite graph $G = (N, M, E)$ where edges are given by:
    \[
    E = \left\{ (i, j) \in N \times M : u_i(j) \geq \tau_j \right\}
    \]
    where $\tau_j := 1 - \frac{1.1 \log{n}}{\alpha_j n}$ and $\tau := \min_{j \in M} \tau_j = 1 - \frac{1.1 \log{n}}{\alpha_{min} n}$.

    Let $\tau' := 1 - \frac{1}{n^{4/x + \varepsilon}}$ for some $\varepsilon > 0$ (say $\epsilon = 0.1$) and $x \geq 5$. The graph $G$ is an Erdős-Rényi random graph with edge probabilities at least $\frac{\log{n} + \omega(1)}{n}$, so by Lemma \ref{lem:Perfect-r-Matching} with high probability, there exists a perfect $x$-matching corresponding to an allocation $\{A_1, \dots, A_n\}$ such that:
    \[
    u_i(A_i) \geq x\tau
    \]

    Now fix any pair $i, i' \in N$. The probability that more than $\frac{x}{2}$ goods in $A_i$ have $u_{i'}(j) > \tau'$ is bounded by:
    \begin{align*}
    \Pr\left( |\{j \in A_i : u_{i'}(j) > \tau'\}| > \frac{x}{2} \right) 
    & \leq \left(\frac{\beta_{max}}{n^{4/x + \varepsilon}} \right)^{x/2} \\ 
    & = o(n^{-2})
    \end{align*}

    By union bound over all $i, i' \in N$:
    \begin{align*}
    & \Pr \left( \exists i,i' \in N: |\{j \in A_i : u_{i'}(j) > \tau'\}| > \frac{x}{2} \right) \\ 
    & = n^2 \cdot o(n^{-2}) = o(1)
    \end{align*}

    Hence, with high probability:
    \[
    u_i(A_{i'}) \leq \frac{x}{2} \tau' + \frac{x}{2}
    \Rightarrow u_i(A_i) - u_i(A_{i'}) \geq x \tau - \left( \frac{x}{2} \tau' + \frac{x}{2} \right)
    \]
    \[
    = \frac{x}{2} \left( \frac{1}{n^{4/x + \varepsilon}} - \frac{2.2 \log{n}}{\alpha_{min} n} \right)
    \]

    We want this to be at least $0$. Thus,
    \[
    \frac{x}{2} \left( \frac{1}{n^{4/x + \varepsilon}} - \frac{2.2 \log{n}}{\alpha_{min} n} \right) \geq 0
    \]
    
    For sufficiently large $n$, and $x$ strictly greater than $4$, $\frac{1}{n^{4/x + \varepsilon}} \geq \frac{2.2 \log{n}}{\alpha_{min} n}$.
    
    % Hence, it is sufficient that
    % \[
    % \frac{x}{4} \cdot \frac{1}{n^{4/x + \varepsilon}} \geq \frac{1.1 \log{n}}{\alpha_{min} n}
    % \Rightarrow x \geq \frac{4.4 \log{n}}{\alpha_{min} n^{1 - 4/x - \varepsilon}}
    % \]

    % Therefore, it suffices to choose $x \geq 5$ (inequality holds for sufficiently large $n$: $\frac{\log{n}}{n^{1 - \frac{4}{x} - \epsilon}}$ gets increasingly smaller with $n$ as $1 - \frac{4}{x} - \epsilon > 0$)

    % Now construct a second graph $G' = (N, M \setminus M_{\leq nx}, E')$ where:
    % \[
    % E' = \left\{ (i, j) \in N \times (M \setminus M_{\leq nx} : u_i(j) \leq \frac{1.1 \log{n}}{\alpha_{min} n} \right\}
    % \]
    % Then $G'$ is again an Erdős-Rényi graph with edge probabilities at least $\frac{\log{n} + \omega(1)}{n}$, so by Lemma \ref{lem:Right-Saturated-Matching} with high probability, there exists a right-saturated matching corresponding to allocation $\{A_1^1, \dots, A_n^1\}$.

    The allocation $(A_1, \dots, A_n)$ is thus envy-free.
\end{proof}

%% file: 6.envy-free-sampling.tex
\section{Sampling with Large Number of Goods}\label{sec:ef-sample}

The model for this result assumes a discrete distribution with a finite support for the utilities of goods. For each $j \in M$, there exists a distribution $\mathcal{U}_j$ with finite support (say $\Omega_j$) on $[0,1]$, such that $u_i(j) \sim \mathcal{U}_j$ independently for each $i \in N$. Each $\mathcal{U}_j$ is $(\alpha_j, \beta_j)$-PDF bounded ($\alpha_j , \beta_j$ may be distinct for each $j \in M$) and for all $j \in M$, $1 \in \Omega_j$ (the value $1$ is always in the support of the distribution). Let  $\mu_j = \mathbb{E}[\mathcal{U}_j]$.

\samplingdiscrete*

% \begin{theorem}
%     Let \( m = \Omega \left(n \log{n} \right) \) and the number of sampled agents for each good \( s = \frac{2 \log{m}}{\alpha_{min}} \). Then a Pareto optimal and EF (envy-free) online allocation of goods is possible with high probability as \( n \to \infty \).
% \end{theorem}

As a direct corollary, we also have

\begin{corollary}\label{cor:sampling-discrete}
    Let \( m = \Omega \left(n \log{n} \right) \) and \( m = \text{poly} \left( n \right)\), and the number of sampled agents for each good \( s = \Theta \left( \log{n} \right) \). Then, with high probability, a maximum social welfare and envy-free online allocation of goods exists as \( n \to \infty \).
\end{corollary}

\begin{proof}
    First, let us consider the maximum value for any good $j \in M$ from the sampled set $S_j$ of agents. For any good $j$

    \begin{equation*}
        \Pr \left( \forall i \in S_j, u_i(j) \neq 1 \right) \leq (1 - \alpha_{j})^{s}
    \end{equation*}

    By union bound over all goods,

    \begin{align*}
        \Pr \left( \exists j \text{ s.t. } \forall i \in S_j, u_i(j) \neq 1 \right) 
        & = \sum_{j \in M} (1 - \alpha_{j})^{s} \\ 
        & \leq m (1 - \alpha_{min})^{s} \\
        & \leq m \cdot e^{-\alpha_{min}s} \\
        & \leq \frac{1}{m}
    \end{align*}

    where the last inequality is due to the chosen value of $s$.

    Thus, with high probability, for all goods, there is always at least one agent in the sampled set of agents that has value equal to $1$. Hence, the allocation has maximum social welfare by construction. 
    
    If good $j$ is allocated to agent $i$ (that is $j \in A_i$), we also know that with high probability $u_i(j) = 1$. This in turn implies that, with high probability, $\mathbb{E} \left[ u_i(j) \middle| j \in A_i \right] = 1$

    Consider any distribution $\mathcal{U}_j$. Let the second largest element in $\Omega_j$ (after $1$) be $1 - \delta_j$ for some $\delta_j > 0$. Now,

    \begin{align*}
        \mu_j 
        & \leq \beta \cdot 1 + \left( 1 - \beta \right) \cdot \left( 1 - \delta_j \right) \\
        & = 1 - \delta + \beta \cdot \delta \\
        & < 1
    \end{align*}

    where the last inequality follows from the fact that $\beta < 1 \implies \beta \cdot \delta < \delta$.

    We also have

    \begin{align*}
        \mu_j 
        & = \mathbb{E} \left[ u_i(j) \middle| j \in A_i \right] \cdot \Pr \left( j \in A_i \right) + \\ 
        & \mathbb{E} \left[ u_i(j) \middle| j \notin A_i \right] \cdot \Pr \left( j \notin A_i \right) \\
        & = 1 \cdot \frac{1}{n} + \mathbb{E} \left[ u_i(j) \middle| j \notin A_i \right] \cdot \frac{n - 1}{n}
    \end{align*}

    ($\Pr \left( j \in A_i \right) = \frac{1}{n}$ as each agent is equally likely to be in $S_j$ and have $u_i(j) = 1$).

    \begin{align*}
        \mathbb{E} \left[ u_i(j) \middle| j \notin A_i \right]
        & = \frac{n \mu_j - 1}{n - 1} \\
        & = \frac{n \mu_j - \mu_j}{n - 1} + \frac{\mu_j - 1}{n - 1} \\
        & = \mu_j - \frac{1 - \mu_j}{n - 1} \\
        & < \mu_j
    \end{align*}

    where the last inequality follows from $\mu_j < 1$.

    We now show that the allocation is EF with high probability.
    
    Define for fixed \( i \in N \), \( j \in M \):
    
    \[
    X_i^j =
    \begin{cases}
    u_i(j), & \text{if } j \in A_i \\
    0, & \text{otherwise}
    \end{cases}
    \]
    
    Then \( u_i(A_i) = \sum_{j \in M} X_i^j \). Moreover,
    
    \begin{align*}
    \mathbb{E}[X_i^j] & = \Pr\left[ j \in A_i \right] \cdot  \mathbb{E}\left[ u_i(j) \,\middle|\, j \in A_i \right] \\
    & = \frac{1}{n} \cdot 1 = \frac{1}{n}
    \end{align*}
    
    Using linearity of expectation,
    
    \begin{align*}
    \mathbb{E}[u_i(A_i)] = \sum_{j \in M} \mathbb{E}[X_i^j] = \frac{m}{n}
    \end{align*}
    
    Now, define for \( i, i' \in N \), \( i \neq i' \), and \( j \in M \):
    
    \[
    Y_{ii'}^j =
    \begin{cases}
    u_i(j), & \text{if } j \in A_{i'} \\
    0, & \text{otherwise}
    \end{cases}
    \]
    
    Then \( u_i(A_{i'}) = \sum_{j \in M} Y_{ii'}^j \). Furthermore,
    
    \begin{align*}
    \mathbb{E}[Y_{ii'}^j] 
    &= \frac{1}{n} \cdot \mathbb{E}\left[ u_i(j) \,\middle|\, j \in A_{i'} \right] \\
    &= \frac{1}{n} \cdot \mathbb{E}\left[ u_i(j) \,\middle|\, j \notin A_i \right] \leq \frac{\mu_j}{n}
    \end{align*}
    
    Let \( Z_{ii'}^j \in [0,1] \) be such that \( \mathbb{E}[Z_{ii'}^j] = \frac{\mu_j}{n} \), and \( Z_{ii'}^j \) stochastically dominates \( Y_{ii'}^j \). Then:
    
    \[
    \sum_{j \in M} \mathbb{E}[Z_{ii'}^j] = \mu_a \cdot \frac{m}{n}, \quad \text{where } \mu_a = \frac{1}{m} \sum_{j \in M} \mu_j
    \]
    
    and for all \( x \in \mathbb{R}^+ \):
    
    \[
    \Pr\left[\sum_{j \in M} Z_{ii'}^j \geq x\right] \geq \Pr\left[\sum_{j \in M} Y_{ii'}^j \geq x\right]
    \]
    
    Let \( E_{ii'} \) be the event that agent \( i \) envies agent \( i' \), i.e., \( \sum_{j \in M} Y_{ii'}^j > \sum_{j \in M} X_i^j \). This occurs only if:
    
    \begin{align*}
    \sum_{j \in M} X_i^j 
    & \leq \frac{m}{n} - \frac{1 - \mu_a}{2} \frac{m}{n} \\ 
    & = \left(1 - \frac{1 - \mu_a}{2}\right) \frac{m}{n} \\ 
    & \leq \left(1 - \frac{1 - \mu_a}{2}\right) \cdot \mathbb{E}\left[\sum_{j \in M} X_i^j\right]
    \end{align*}
    
    or
    
    \begin{align*}
    \sum_{j \in M} Y_{ii'}^j 
    & \geq \mu_a \frac{m}{n} + \frac{1 - \mu_a}{2} \frac{m}{n} \\ 
    & = \left(1 + \frac{1 - \mu_a}{2\mu_a}\right) \mu_a \frac{m}{n} \\
    & = \left(1 + \frac{1 - \mu_a}{2\mu_a}\right) \cdot \mathbb{E}\left[\sum_{j \in M} Z_{ii'}^j\right]
    \end{align*}
    
    Let \( \epsilon = \min\left\{1, \frac{1 - \mu_a}{2} \right\} \), so \( \epsilon < \frac{1 - \mu_a}{2\mu_a} \) also holds. $\epsilon > 0$ also holds as $\mu_j < 1$ for all $j$ implies $\mu_a < 1$.
    
    Since \( X_i^j \) and \( Z_{ii'}^j \) are independent across \( j \), using Chernoff bounds:
    
    \begin{align*}
    & \Pr\left[ \sum_{j \in M} X_i^j \leq (1 - \epsilon) \cdot \mathbb{E}\left[\sum_{j \in M} X_i^j\right] \right] \\
    & \leq \exp\left(- \frac{\epsilon^2}{2} \cdot \frac{m}{n} \right)
    \end{align*}
    
    and
    
    \begin{align*}
    & \Pr\left[ \sum_{j \in M} Y_{ii'}^j \geq (1 + \epsilon) \cdot \mathbb{E}\left[\sum_{j \in M} Z_{ii'}^j\right] \right] \\ 
    & \leq \exp\left(- \frac{\epsilon^2}{3} \mu_a \cdot \frac{m}{n} \right)
    \end{align*}
    
    Set \( n \leq \frac{\epsilon^2 \mu_a}{3} \cdot \frac{m}{\ln(2m^3)} \). By the union bound:
    
    \begin{align*}
    \Pr[E_{ii'}] 
    &\leq \exp\left(- \frac{\epsilon^2}{2}  \cdot \frac{m}{n} \right) + \exp\left(- \frac{\epsilon^2}{3} \mu_a \cdot \frac{m}{n} \right) \\
    &\leq 2 \cdot \frac{1}{2m^3} = \frac{1}{m^3}
    \end{align*}
    
    Allocation \( A \) is EF if and only if no event \( E_{ii'} \) occurs. The probability that \( A \) is not EF is at most:
    
    \[
    \Pr\left[\bigvee_{\substack{i, i' \in N \\ i \neq i'}} E_{ii'}\right] \leq \sum_{\substack{i, i' \in N \\ i \neq i'}} \Pr[E_{ii'}] \leq \binom{n}{2} \cdot \frac{1}{m^3} \leq \frac{1}{m}
    \]
    
    Thus, the probability that \( A \) is not EF vanishes as \( m \to \infty \) (\( n \to \infty \) and \(m = \Omega (n \log{n})\)).
\end{proof}

We present similar, but weaker results for continuous distributions. For each $j \in M$, there exists a distribution $\mathcal{U}_j$ supported on $[0,1]$, such that $u_i(j) \sim \mathcal{U}_j$ independently for each $i \in N$. Each $\mathcal{U}_j$ is \textit{non-atomic}, \textit{continuous} and $(\alpha_j, \beta_j)$-PDF bounded ($\alpha_j , \beta_j$ may be distinct for each $j \in M$).

\samplingcontinuousconstant*

% \begin{theorem}
%     Let \( m = \Omega \left(n \log{n} \right) \) and the number of sampled agents for each good \( s = \frac{20 \log{m}}{\alpha_{min}} \). Then a $0.9$ Pareto optimal and $0.8$ EF (envy-free) online allocation of goods is possible with high probability as \( n \to \infty \).
% \end{theorem}

As a direct corollary, we also have

\begin{corollary}\label{cor:sampling-continuous-constant}
    Let \( m = \Omega \left(n \log{n} \right) \) and \( m = \text{poly} \left( n \right)\), and the number of sampled agents for each good \( s = \Theta \left( \log{n} \right) \). Then, with high probability, a $0.9$-approx. maximum social welfare and $0.8$-approx. envy-free online allocation of goods exists as \( n \to \infty \).
\end{corollary}

\begin{proof}
    First, let us consider the maximum value for any good $j \in M$ from the sampled set $S_j$ of agents. For any good $j$

    \begin{equation*}
        \Pr \left( \forall i \in S_j, u_i(j) < 0.9 \right) \leq (1 - 0.1 \alpha_{j})^{s}
    \end{equation*}

    By union bound over all goods,

    \begin{align*}
        \Pr \left( \exists j \text{ s.t. } \forall i \in S_j, u_i(j) < 0.9 \right) 
        & = \sum_{j \in M} (1 - 0.1 \alpha_{j})^{s} \\ 
        & \leq m (1 - 0.1 \alpha_{min})^{s} \\
        & \leq m \cdot e^{- 0.1 \alpha_{min}s} \\
        & \leq \frac{1}{m}
    \end{align*}

    where the last inequality is due to the chosen value of $s$.

    Thus, with high probability, for all goods, there is always at least one agent in the sampled set of agents that has value greater than or equal to $0.9$. Hence, the allocation is $0.9$-approximate maximum social welfare, with high probability, by construction. 
    
    If good $j$ is allocated to agent $i$ (that is $j \in A_i$), we also know that with high probability $u_i(j) \geq 0.9$. This in turn implies that, with high probability, $\mathbb{E} \left[ u_i(j) \middle| j \in A_i \right] \geq 0.9$.

    We now show that the allocation is $0.8$ EF with high probability.
    
    Define for fixed \( i \in N \), \( j \in M \):
    
    \[
    X_i^j =
    \begin{cases}
    u_i(j), & \text{if } j \in A_i \\
    0, & \text{otherwise}
    \end{cases}
    \]
    
    Then \( u_i(A_i) = \sum_{j \in M} X_i^j \). Moreover,
    
    \begin{align*}
    \mathbb{E}[X_i^j] & = \Pr\left[ j \in A_i \right] \cdot  \mathbb{E}\left[ u_i(j) \,\middle|\, j \in A_i \right] \\
    & \geq \frac{1}{n} \cdot 0.9 = \frac{0.9}{n}
    \end{align*}
    
    Using linearity of expectation,
    
    \begin{align*}
    \mathbb{E}[u_i(A_i)] = \sum_{j \in M} \mathbb{E}[X_i^j] \geq \frac{0.9 m}{n}
    \end{align*}
    
    Now, define for \( i, i' \in N \), \( i \neq i' \), and \( j \in M \):
    
    \[
    Y_{ii'}^j =
    \begin{cases}
    u_i(j), & \text{if } j \in A_{i'} \\
    0, & \text{otherwise}
    \end{cases}
    \]
    
    Then \( u_i(A_{i'}) = \sum_{j \in M} Y_{ii'}^j \). Furthermore,
    
    \begin{align*}
    \mathbb{E}[Y_{ii'}^j] 
    &= \frac{1}{n} \cdot \mathbb{E}\left[ u_i(j) \,\middle|\, j \in A_{i'} \right] \\
    &= \frac{1}{n} \cdot \mathbb{E}\left[ u_i(j) \,\middle|\, j \notin A_i \right] \leq \frac{1}{n}
    \end{align*}

    where the last inequality holds as all utilities are less than or equal to $1$.
    
    Let \( Z_{ii'}^j \in [0,1] \) be such that \( \mathbb{E}[Z_{ii'}^j] = \frac{1}{n} \), and \( Z_{ii'}^j \) stochastically dominates \( Y_{ii'}^j \). Then:
    
    \[
    \sum_{j \in M} \mathbb{E}[Z_{ii'}^j] = \frac{m}{n}
    \]
    
    and for all \( x \in \mathbb{R}^+ \):
    
    \[
    \Pr\left[\sum_{j \in M} Z_{ii'}^j \geq x\right] \geq \Pr\left[\sum_{j \in M} Y_{ii'}^j \geq x\right]
    \]
    
    Let \( E_{ii'} \) be the event that \( \frac{u_i \left( A_i \right) }{u_i \left( A_{i'} \right) } < 0.8 \), i.e., \( 0.8 \sum_{j \in M} Y_{ii'}^j > \sum_{j \in M} X_i^j \). This occurs only if:
    
    \begin{align*}
    \sum_{j \in M} X_i^j 
    & \leq 0.85 \frac{m}{n} \\ 
    & = \left(1 - \frac{0.05}{0.9}\right) \frac{0.9 m}{n} \\ 
    & \leq \left(1 - \frac{0.05}{0.9}\right) \cdot \mathbb{E}\left[\sum_{j \in M} X_i^j\right]
    \end{align*}
    
    or
    
    \begin{align*}
    \sum_{j \in M} Y_{ii'}^j 
    & \geq \frac{0.85}{0.8} \cdot \frac{m}{n} \\ 
    & = \left(1 + \frac{0.05}{0.8}\right) \frac{m}{n} \\
    & = \left(1 + \frac{0.05}{0.8}\right) \cdot \mathbb{E}\left[\sum_{j \in M} Z_{ii'}^j\right]
    \end{align*}
    
    Let \( \epsilon = \frac{0.05}{0.9} \), so \( \epsilon < \frac{0.05}{0.8} \) also holds.
    
    Since \( X_i^j \) and \( Z_{ii'}^j \) are independent across \( j \), using Chernoff bounds:
    
    \begin{align*}
    & \Pr\left[ \sum_{j \in M} X_i^j \leq (1 - \epsilon) \cdot \mathbb{E}\left[\sum_{j \in M} X_i^j\right] \right] \\
    & \leq \exp\left(- \frac{\epsilon^2}{2} \cdot \frac{0.9 m}{n} \right)
    \end{align*}
    
    and
    
    \begin{align*}
    & \Pr\left[ \sum_{j \in M} Y_{ii'}^j \geq (1 + \epsilon) \cdot \mathbb{E}\left[\sum_{j \in M} Z_{ii'}^j\right] \right] \\ 
    & \leq \exp\left(- \frac{\epsilon^2}{3} \cdot \frac{m}{n} \right)
    \end{align*}
    
    Set \( n \leq \frac{\epsilon^2}{3} \cdot \frac{m}{\ln(2m^3)} \). By the union bound:
    
    \begin{align*}
    \Pr[E_{ii'}] 
    &\leq \exp\left(- \frac{\epsilon^2}{2}  \cdot \frac{0.9m}{n} \right) + \exp\left(- \frac{\epsilon^2}{3} \cdot \frac{m}{n} \right) \\
    &\leq 2 \cdot \frac{1}{2m^3} = \frac{1}{m^3}
    \end{align*}
    
    Allocation \( A \) is $0.8$-approximate EF if and only if no event \( E_{ii'} \) occurs. The probability that \( A \) is not $0.8$-approximate EF is at most:
    
    \[
    \Pr\left[\bigvee_{\substack{i, i' \in N \\ i \neq i'}} E_{ii'}\right] \leq \sum_{\substack{i, i' \in N \\ i \neq i'}} \Pr[E_{ii'}] \leq \binom{n}{2} \cdot \frac{1}{m^3} \leq \frac{1}{m}
    \]
    
    Thus, the probability that \( A \) is not $0.8$-approximate EF vanishes as \( m \to \infty \) (\( n \to \infty \) and \(m = \Omega (n \log{n})\)).
\end{proof}

%% file: 7.envy-free-large-main.tex
% \section{Sampling with Large Number of Goods}\label{sec:large-ef-sample}

For our next result, we will assume the online and sampling model as follows. We assume that each good arrives one-by-one (online), and must be allocated immediately and irrevocably. For each good, we sample $s$ agents at random (uniformly). Let $S_j \subset N$ be the $s$ agents sampled for good $j \in M$. Out of these $s$ agents, we assign the good to agent with maximum value for the good. We will show that this simple algorithm gives strong envy-free and maximum social welfare guarantees. Our model is as discussed in Section \ref{sec:intro-model}, with the following additional mild constraint. As long as we are guaranteed a constant upper bound (say $\mu < 1$) on the means of the distributions of the goods, that is, $\mu_j < 1$ for all goods $j$, we have the following result.

\samplingcontinuous*

As a direct corollary, we also have,

\begin{corollary}\label{cor:sampling-continuous}
    Let \( m = \Omega (n \log{n}) \) and \( m = \text{poly} ( n )\), and the number of sampled agents for each good \( s = \Theta (( \log{n})^2) \). If, for all goods $j$, $\mathbb{E}[ \mathcal{U}_j] \leq \mu$ for some constant $\mu < 1$, then, with high probability, a $( 1 - \frac{1}{\log {n}})$-approx. maximum social welfare and envy-free online allocation of goods exists as \( n \to \infty \).
\end{corollary}

\begin{proof}
    For any good $j$,
    \begin{equation*}
        \Pr \left( \forall i \in S_j, u_i(j) < 1 - \frac{1}{\log{n}} \right) \leq \left(1 - \frac{\alpha_{j}}{\log{n}} \right)^{s}
    \end{equation*}
    By union bound over all goods,
    \begin{align*}
        \Pr & \left( \exists j \text{ s.t. } \forall i \in S_j, u_i(j) < 1 - \frac{1}{\log{n}} \right) \\
        & \leq \sum_{j \in M} \left(1 - \frac{\alpha_{j}}{\log{n}} \right)^{s} \leq m \left(1 - \frac{\alpha_{min}}{\log{n}} \right)^{s} \\
        & \leq m \cdot e^{- \frac{\alpha_{min}}{\log{n}} \cdot s}
        \leq \frac{1}{m}
    \end{align*}
    where the last inequality is due to the chosen value of $s$.

    Thus, with high probability, for all goods, there is always at least one agent in the sampled set of agents that has value greater than or equal to $1 - \frac{1}{\log{n}} \geq \frac{\mu + 1}{2}$ (for sufficiently large $n$). Hence, the allocation is $(1 - \frac{1}{\log {n}})$- maximum social welfare by construction. 
    
    If good $j$ is allocated to agent $i$ (that is $j \in A_i$), we also know that with high probability $u_i(j) \geq \frac{\mu + 1}{2}$. This in turn implies that, with high probability, $\mathbb{E} \left[ u_i(j) \middle| j \in A_i \right] \geq \frac{\mu + 1}{2}$.

    We now show that the allocation is EF with high probability.
    
    Define for fixed \( i \in N \), \( j \in M \):
    \[
    X_i^j =
    \begin{cases}
    u_i(j), & \text{if } j \in A_i \\
    0, & \text{otherwise}
    \end{cases}
    \]
    Then \( u_i(A_i) = \sum_{j \in M} X_i^j \). Moreover,
    \begin{align*}
    \mathbb{E}[X_i^j] & = \Pr\left[ j \in A_i \right] \cdot  \mathbb{E}\left[ u_i(j) \,\middle|\, j \in A_i \right] \geq \frac{1}{n} \cdot \frac{\mu + 1}{2}
    \end{align*}
    Using linearity of expectation,
    \begin{align*}
    \mathbb{E}[u_i(A_i)] = \sum_{j \in M} \mathbb{E}[X_i^j] \geq \frac{\mu + 1}{2} \cdot \frac{m}{n}
    \end{align*}
    Now, define for \( i, i' \in N \), \( i \neq i' \), and \( j \in M \):
    \[
    Y_{ii'}^j =
    \begin{cases}
    u_i(j), & \text{if } j \in A_{i'} \\
    0, & \text{otherwise}
    \end{cases}
    \]
    Then \( u_i(A_{i'}) = \sum_{j \in M} Y_{ii'}^j \). Furthermore,
    \begin{align*}
    \mathbb{E}[Y_{ii'}^j] 
    &= \frac{1}{n} \cdot \mathbb{E}\left[ u_i(j) \,\middle|\, j \in A_{i'} \right] \\
    &= \frac{1}{n} \cdot \mathbb{E}\left[ u_i(j) \,\middle|\, j \notin A_i \right] \leq \frac{\mu}{n}
    \end{align*}
    (since $\mathbb{E}\left[ u_i(j) \,\middle|\, j \notin A_i \right] \leq \mu_j \leq \mu$, where the first inequality holds from the proof of Theorem \ref{thm:sampling-discrete}, shown in appendix \ref{sec:ef-sample})
    
    Let \( Z_{ii'}^j \in [0,1] \) be such that \( \mathbb{E}[Z_{ii'}^j] = \frac{\mu}{n} \), and \( Z_{ii'}^j \) stochastically dominates \( Y_{ii'}^j \). Then, $\sum_{j \in M} \mathbb{E}[Z_{ii'}^j] = \mu \cdot \frac{m}{n}$.
    
    % and for all \( x \in \mathbb{R}^+ \):
    
    % \[
    % \Pr\left[\sum_{j \in M} Z_{ii'}^j \geq x\right] \geq \Pr\left[\sum_{j \in M} Y_{ii'}^j \geq x\right]
    % \]
    
    Let \( E_{ii'} \) be the event that agent \( i \) envies agent \( i' \), i.e., \( \sum_{j \in M} Y_{ii'}^j > \sum_{j \in M} X_i^j \). This occurs only if either of the two conditions below occur:
    \begin{align*}
    \sum_{j \in M} X_i^j 
    & \leq \frac{3 \mu + 1}{4} \cdot \frac{m}{n} = \left( 1 - \frac{1 - \mu}{ 2 \left( \mu + 1 \right)} \right) \frac{\mu + 1}{2} \cdot \frac{m}{n} \\ 
    & \leq \left( 1 - \frac{1 - \mu}{ 2 \left( \mu + 1 \right)} \right) \cdot \mathbb{E}\left[\sum_{j \in M} X_i^j\right] \\
    \sum_{j \in M} Y_{ii'}^j 
    & \geq \frac{3 \mu + 1}{4} \cdot \frac{m}{n} = \left(1 + \frac{1 - \mu}{4 \mu} \right) \mu \cdot \frac{m}{n} \\
    & = \left(1 + \frac{1 - \mu}{4 \mu} \right) \cdot \mathbb{E}\left[\sum_{j \in M} Z_{ii'}^j\right]
    \end{align*}
    Let \( \epsilon = \frac{1 - \mu}{2 \left( \mu + 1 \right)} < 1 \), so \( \epsilon < \frac{1 - \mu}{4 \mu} \) also holds.
    
    Since \( X_i^j \) and \( Z_{ii'}^j \) are independent across \( j \), using Chernoff bounds:
    \begin{align*}
    & \Pr\left[ \sum_{j \in M} X_i^j \leq (1 - \epsilon) \cdot \mathbb{E}\left[\sum_{j \in M} X_i^j\right] \right] \\
    & \leq \exp\left(- \frac{\epsilon^2}{2} \cdot \frac{\mu + 1}{2} \cdot \frac{m}{n} \right) \\
    & \Pr\left[ \sum_{j \in M} Y_{ii'}^j \geq (1 + \epsilon) \cdot \mathbb{E}\left[\sum_{j \in M} Z_{ii'}^j\right] \right] \\ 
    & \leq \exp\left(- \frac{\epsilon^2}{3} \cdot \mu \cdot \frac{m}{n} \right)
    \end{align*}
    Choose $m$ large enough such that \( n \leq \frac{\epsilon^2}{3} \cdot \mu \cdot \frac{m}{\ln(2m^3)} \) (since \( m = \Omega(n \log{n})\)). By the union bound:
    \begin{align*}
    \Pr[E_{ii'}] 
    &\leq \exp\left(- \frac{\epsilon^2}{2} \cdot \frac{\mu + 1}{2} \cdot \frac{m}{n} \right) + \exp\left(- \frac{\epsilon^2}{3} \cdot \mu \cdot \frac{m}{n} \right) \\
    &\leq 2 \cdot \frac{1}{2m^3} = \frac{1}{m^3}
    \end{align*}
    Allocation \( A \) is EF if and only if no event \( E_{ii'} \) occurs. The probability that \( A \) is not EF is at most:
    \[
    \Pr\left[\bigvee_{\substack{i, i' \in N \\ i \neq i'}} E_{ii'}\right] \leq \sum_{\substack{i, i' \in N \\ i \neq i'}} \Pr[E_{ii'}] \leq \binom{n}{2} \cdot \frac{1}{m^3} \leq \frac{1}{m}
    \]
    Thus, the probability that \( A \) is not EF vanishes as \( m \to \infty \) (\( n \to \infty \) and \(m = \Omega (n \log{n})\)).
\end{proof}

%% file: 8.prop-small.tex
\section{Proportionality for Small Number of Goods}\label{sec:prop-small}

%\subsection{Small number of goods $(m = O(n))$}

% \begin{algorithm}[t]
% \caption{Proportional Algorithm for $m = n$}
% \label{alg:Threshold Matching}
% \begin{algorithmic}[1]
% \Procedure{TauMatch$_\tau$}{$N, M, \{u_i\}_{i \in N}$}
%     \State Let $E_{\geq \tau} \gets \{(i, j) \in N \times M \mid u_i(j) \geq \tau_j\}$
%     \If{$G_{\geq \tau} = (N, M, E_{\geq \tau})$ contains a perfect matching}
%         \State \Return any perfect matching of $G_{\geq \tau}$
%     \Else
%         \State \Return NULL
%     \EndIf
% \EndProcedure
% \end{algorithmic}
% \end{algorithm}

\begin{algorithm}[t]
\caption{Proportional Algorithm for $n \leq m \leq 2n$}
\label{alg:Two-step algorithm}
\begin{algorithmic}[1]
\State \textbf{Input:} $N, M, \{u_i\}_{i \in N}$
\State $M^0\coloneqq$ first $n$ items in $M$, $M^1 \coloneqq M \setminus M^0$.
\State $\tau_j \gets 1-\frac{1.1\log n}{\alpha_j n},\quad $ $\forall j\in M$
\State $E_{\geq\tau} \gets \{(i, j) \in N \times M^0 \mid u_i(j) \geq \tau_j\}$
\If{$G_{\geq \tau}(N, M^0, E_{\geq\tau})$ has perfect matching}
    \State $(M_1^0, \ldots, M_n^0) \gets$ any perfect matching of $G_{\geq\tau}$
    \State $N^{\text{vio}} \gets \{ i \in N \mid u_i(M^0_i) < \frac{u_i(M)}{n} \}$
    \State $E^{\text{fix}} \gets \{ (i, j) \mathsmaller{\in N^{\text{vio}} \times M^1 \mid u_i(j) \geq \frac{u_i(M)}{n} - u_i(M^0_i)} \}$
    \If{$G^{\text{fix}}(\mathsmaller{N^{\text{vio}}, M^1, E^{\text{fix}}})$ has $|N^{\text{vio}}|$ size matching}
        \State $(M^1_1, \dots, M^1_n) \gets$ any such matching
        \State \Return $(M^0_1 \cup M^1_1, \dots, M^0_n \cup M^1_n)$
    \EndIf
\EndIf
\end{algorithmic}
\end{algorithm}

\propsmallgoods*

% \begin{theorem}
%     For allocating goods, if $n \le m \le 2n,$ a proportional allocation always exists as $n\rightarrow \infty,$ if and only if for all $j,$ $\mathcal{U}_j$ has mean at most $\nicefrac{1}{2}$.
% \end{theorem}

The assumption that each distribution \(\mathcal{U}_j\) has a mean at most \(1/2\) is crucial to ensure the existence of a proportional allocation for all \(n \leq m \leq 2n\) with high probability. Suppose instead \(\mathcal{U}_j\) has mean \(1/2 + \varepsilon_j\) for some constant \(\varepsilon_j > 0\), and set \(m = 2n - 1\). Then, the expected utilities satisfy
$
\mathbb{E}[u_i(M)] = n \left(1 + \frac{1}{n}\sum_j \varepsilon_j - o(1)\right)
$.
By standard Chernoff and union bounds, with high probability, all agents simultaneously have \(u_i(M) > n\). In such cases, no proportional allocation exists since at least one agent must receive at most one item (by the pigeonhole principle), but each item has value at most 1, contradicting proportionality.

% Before we begin with the proof of Theorem~\ref{thm:prop-small-m}, let us try to see Algorithm~\ref{alg:Threshold Matching} with respect to Theorem~\ref{thm:Extension to Erdos Renyi}. If we set $\tau_j = 1 - \frac{1.1 \log n}{\alpha_j \, n}$ for all $j \in M$ in our setting, then Algorithm~\ref{alg:Threshold Matching} returns a perfect matching with high probability when $n$ is large.

\begin{proof}
    Consider the graph $G_{\geq\tau} = \{N, M^0, E_{\geq\tau}\}$ defined in the first stage of Algorithm~\ref{alg:Two-step algorithm} with $\tau_j = 1 - \frac{1.1 \log n}{\alpha_j n}$. Since Theorem~\ref{thm:Extension to Erdos Renyi} immediately implies that a perfect matching exists in $G_{\geq\tau}$ with high probability, we will assume that this is the case from now on. Also, using Lemma~\ref{lem:Hoeffding's inequality}, we can show that
    \begin{equation}\label{eq:u_i(M) bound}
    u_i(M) \leq \frac{m}{2} + 10 \sqrt{m \log n}
    \end{equation}
    holds with high probability. Hence, we will also assume this from now on. Now, let $q = m-n$ (so $0 \leq q \leq n$). We make two cases based on the value of $q$.

    \paragraph{Case $1$: $q < 0.9 n.$} For sufficiently large $n$, we have from equation~\eqref{eq:u_i(M) bound}, $\frac{u_i(M)}{n} \leq 0.996$ and for sufficiently large $n$, $\tau_j > 0.996$. Hence, $\tau_j > \frac{u_i(M)}{n}, \forall j \in M$ for all $i \in N$. Hence, running only the first part of the algorithm gives a proportional allocation, i.e., the partial allocation $(M^0_1, \dots, M^0_n)$ is already proportional.

    \paragraph{Case $2$: $q \geq 0.9 n.$} Define $G^1_{\geq \frac{0.5}{\beta}} = (N, M^1, E^1_{\geq \frac{0.5}{\beta}})$ such that $(i, j) \in E^1_{\geq \frac{0.5}{\beta}}$ if and only if $u_i(j) \geq \frac{0.5}{\beta_j}$. Let $Z_G(S)$ denote the set of vertices adjacent to the vertices of the set $S$ in graph $G$. We define two events $E1$ and $E2$ as follows. $E1$ occurs when $|N^{\text{violated}}| \leq 0.6 n,$ and $E2$ occurs when $\forall S \subseteq N$ and $|S| \leq 0.6 n$, $\left|Z_{G^1_{\geq \frac{0.5}{\beta}}}(S) \right| \geq |S|.$

    We first prove that if $E1$ and $E2$ hold with high probability, then the algorithm outputs a proportional allocation. For sufficiently large $n$, $\frac{u_i(M)}{n} = 1 + o(1)$ (using equation~\eqref{eq:u_i(M) bound}) and $1+o(1) < \frac{0.5}{\beta_j} + \tau_j$. Now, if $E1$ and $E2$ hold, due to Hall's Marriage Theorem, there exists a matching from $N^{\text{violated}}$ to $M^1$ that uses all vertices in $N^{\text{violated}}$. And since $\frac{u_i(M)}{n} < \frac{0.5}{\beta_j} + \tau_j$, $\forall i \in N$ and $\forall j \in M$, the matching in $G^1_{\geq \frac{0.5}{\beta}}$ from $N^{\text{violated}}$ to $M^1$ also is a matching in $G^{\text{fix}}$. The algorithm thus outputs a proportional allocation. Thus, proving that $E1$ and $E2$ hold with high probability suffices. We now proceed to show this.

    Let $f_j$ denote the probability distribution function of $\mathcal{U}_j$. Since $f_j(x) \leq \beta_j, \forall x \in [0,1]$, $\Pr\bigl[u_i(j) \geq \nicefrac{0.5}{\beta_j} \bigr] \geq 0.5$. Hence, for $G^1_{\geq \frac{0.5}{\beta}}$, Lemma~\ref{lem:Hall's-condition-satisfying-case}, directly implies that $E2$ holds with high probability if $E1$ holds.

     For showing $E1,$ let $X_j = u_i(j) - \mu_j$. Thus $\mathbb{E}[X_j] = 0$ and $\mathbb{E}[X_j^2] = \sigma_j^2$, where $\sigma_j$ is the standard deviation of the distribution $\mathcal{U}_j$. Also define $\rho_j = \mathbb{E}\bigl[|X_j|^3\bigr]$. Thus, $\mathbb{E}\bigl[X_j^2\bigr] = \int_{-\mu_j}^{1-\mu_j} x^2 f_j(x) dx.$ Since $\alpha_j \leq f_j(x) \leq \beta_j$ for every $x \in [-\mu_j, 1-\mu_j]$,
    \[
    \alpha_j\int_{-\mu_j}^{1-\mu_j} x^2 dx \;\leq\; \mathbb{E}\bigl[X_j^2\bigr] \;\leq\; \beta_j\int_{-\mu_j}^{1-\mu_j} x^2 dx
    \]
    \[
    \alpha_j \frac{1 - 3\mu_j + 3\mu_j^2}{3} \;\leq\; \mathbb{E}\bigl[X_j^2\bigr] \;\leq\; \beta_j\frac{1 - 3\mu_j + 3\mu_j^2}{3}
    \]
    Now, since $\mu_j \in [0, 1/2]$, we have $1/4 \leq 1 - 3\mu_j + 3\mu_j^2 \leq 1$, which means
    \begin{equation}\label{eq:E[X^2] bound}
    \frac{\alpha_j}{12} \;\leq\; \mathbb{E}\bigl[X_j^2\bigr] \;\leq\; \frac{\beta_j}{3}
    \end{equation}
    Now, consider the quantity $\frac{\sum_j \rho_j}{\left(\sum_j \sigma_j^2\right)^2}$. By Lyapunov’s moment inequality (see~\cite{karrNotes}, Corollary 4.110, p.182), we have
    $
    \mathbb{E}[|X_j|^3] \le \mathbb{E}[X_j^2]^{3/2} = \sqrt{\mathbb{E}[X_j^2]}\,\mathbb{E}[X_j^2]
    $, so
    \begin{align*}
        \frac{\sum_j \rho_j}{\left(\sum_j \sigma_j^2\right)^2} &\leq \frac{\sum_j \sigma_j^3}{\left(\sum_j \sigma_j^2\right)^2} \leq \frac{\sigma_{max}\sum_j \sigma_j^2}{\left(\sum_j \sigma_j^2\right)^2}\\
        &= \frac{\max\sqrt{\mathbb{E}\bigl[X_j^2\bigr]}}{\left(\sum_{j=1}^m \mathbb{E}\bigl[X_j^2\bigr]\right)} \leq \frac{\max \sqrt{(\beta_j / 3)}}{\sum_{j=1}^m (\alpha_j / 12)}
    \end{align*}
    where the first inequality follows from Lyapunov's inequality and the third inequality follows from equation~\eqref{eq:E[X^2] bound}. Thus
    \begin{equation}\label{eq:rho-sigma bound}
        \frac{\sum_{j\in M} \rho_j}{\left(\sum_{j \in M} \sigma_j^2\right)^2} \leq \frac{4\sqrt{3\beta_{max}}}{m \,\alpha_{min}}
    \end{equation}
    where $\alpha_{min}$ and $\beta_{max}$ represent $\min_{j\in M}(\alpha_j)$ and $\max_{j\in M}(\beta_j)$ respectively. Now, we know that
    $
        \Pr [u_i(M) \leq n\min(\tau_j)] \geq \Pr [\sum_{j=1}^m u_i(j) \leq \frac{m}{2}\min(\tau_j)]
    $
    because $m \leq 2n$. Rearranging the terms gives us
    $
        \Pr [\frac{u_i(M)}{n} \leq \min(\tau_j)]
        \geq \Pr [\sum_{j=1}^m X_j\leq \frac{m}{2}\min(\tau_j) - \sum_{j=1}^m \mu_j].
    $ Now $\sum \mu_j \leq m/2$, as every distribution $\mathcal{U}_j$ has a mean at most $1/2$, thus,
    \begin{align*}
        \Pr &[u_i(M)/n \leq \min(\tau_j)]\\
        &\geq \Pr [X_1 + \dots + X_m\leq (m/2)\min(\tau_j) - m/2]\\
        % &= \Pr \left[X_1 + \dots + X_m\leq \frac{m}{2}\left(1 - \frac{1.1\log n}{n \, \alpha_{min}}\right) - \frac{m}{2}\right]\\
        &= \Pr \left[X_1 + \dots + X_m\leq -\frac{1.1 m\log n}{2 n \, \alpha_{min}}\right]
    \end{align*}
    Now, let's define $S_m = \frac{X_1 + \dots + X_m}{\sqrt{\sigma_1^2 + \dots + \sigma_m^2}}$ and let $F_S$ denote the cumulative distribution function of $S_m$. Let $t_n = -\frac{1.1 m\log n}{2 n \, \alpha_{min}}\frac{1}{\sqrt{\sigma_1^2 + \dots + \sigma_m^2}}$. Then
    \begin{equation}\label{eq:pr-fs relation}
    \Pr \left[\frac{u_i(M)}{n} \leq \min(\tau_j)\right] \geq F_S(t_n)
    \end{equation}
    and using equation~\eqref{eq:E[X^2] bound}, $t_n \geq -\frac{1.1 m\log n}{2 n \, \alpha_{min}} \frac{1}{\sqrt{m \alpha_{min} / 12}}$ and $t \leq 0$. Note also that $t_n$ tends to $0$ as $n$ gets sufficiently large.

    Now, by \cite{esseen1942liapunoff}, for independent (not necessarily identically distributed) summands, we can say that,
    $\sup_x |F_S(x) - \Phi(x)| \leq C_0 \cdot \frac{\sum_j \rho_j}{(\sum_j \sigma_j^2)^{3/2}},$ where $\Phi$ represents the PDF of the standard normal distribution. Hence, $F_S(t_n) \geq \Phi(t_n) - C_0 \cdot \frac{\sum_j \rho_j}{(\sum_j \sigma_j^2)^{3/2}}.$ Using equation~\eqref{eq:rho-sigma bound},  $F_S(t_n) \geq \Phi(t_n) - C_0 \cdot \frac{4\sqrt{3\beta_{max}}}{m \,\alpha_{min}}.$ Since, $\alpha_{min}$ and $\beta_{max}$ are finite quantities, for sufficiently large $n$ (and thus also $m$), and as $t_n$ tends to $0$, the above equation becomes $F_S(t_n) \geq \Phi(t_n) - 0.04 \geq 0.49 - 0.04$. Thus, using equation~\eqref{eq:pr-fs relation}, we get $\Pr \left[\frac{u_i(M)}{n} \leq \min(\tau_j)\right] \geq 0.45$ for all agents $i \in N$, so each agent has at most $0.55$ probability of being present in $N^{\text{violated}}$. Hence, using standard concentration inequalities, $|N^{\text{violated}}| \leq 0.6 n$ with high probability.

    Thus, the algorithm outputs a proportional allocation with high probability in both cases. This concludes our proof.
\end{proof}

%% file: 11.empirical-results.tex
\section{Empirical Results}
\label{sec:empirical-results}

% \begin{figure}[t]
%     \centering
%     \includegraphics[width=\picturewidth]{images/comparison_chores_beta_uniform_n100.png}
%     \caption{Chores: \texttt{beta\_uniform} utilities, comparison of sampling and non-sampling algorithms.}
%     \label{fig:chores-beta}
%     \Description{Chores: \texttt{beta\_uniform} utilities, comparison of sampling and non-sampling algorithms.}
% \end{figure}

\begin{figure}[t]
    \centering
    \includegraphics[width=\picturewidth]{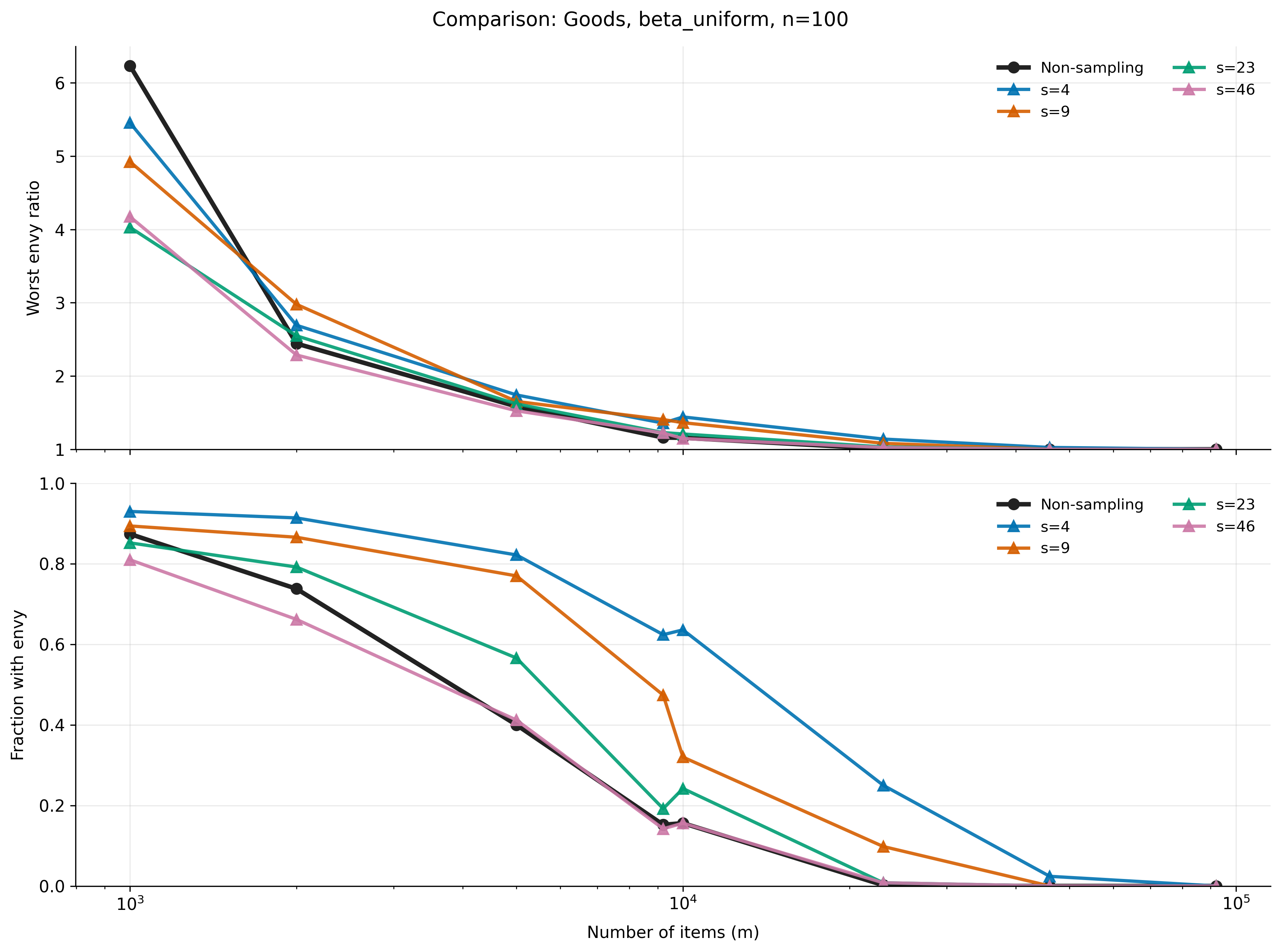}
    \caption{Goods: \texttt{beta\_uniform} utilities, comparison of sampling and non-sampling algorithms.}
    \label{fig:goods-beta}
    % \Description{Goods: \texttt{beta\_uniform} utilities, comparison of sampling and non-sampling algorithms.}
\end{figure}

\begin{figure}[t]
    \centering
    \includegraphics[width=\picturewidth]{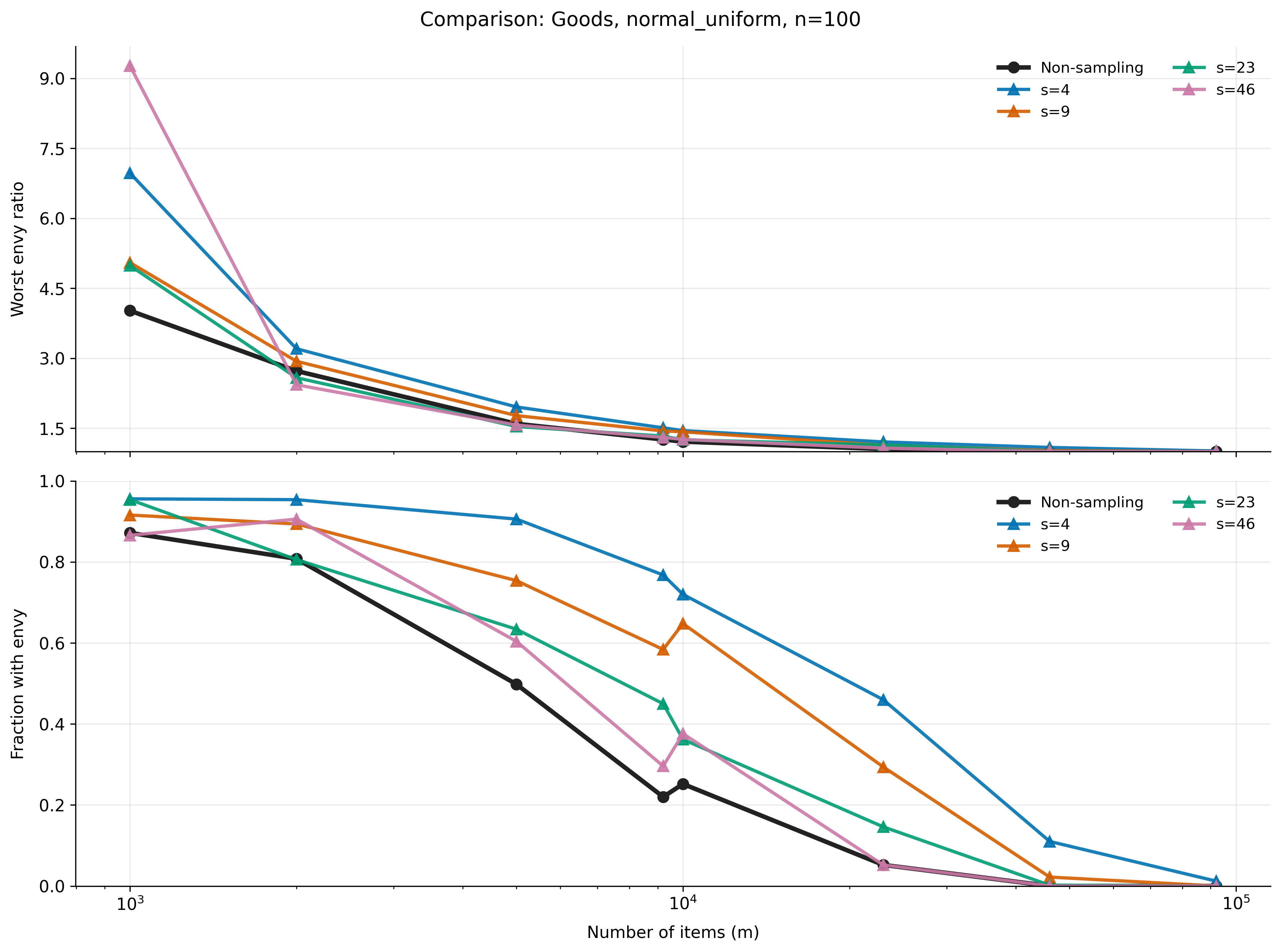}
    \caption{Goods: \texttt{normal\_uniform} utilities, comparison of sampling and non-sampling algorithms.}
    \label{fig:goods-normal}
    % \Description{Goods: \texttt{normal\_uniform} utilities, comparison of sampling and non-sampling algorithms.}
\end{figure}

We empirically evaluate both the non-sampling and sampling algorithms on finite instances to understand how quickly the asymptotic behaviour manifests.\footnote{Appendix~\ref{sec:details-experiments} contains all details on the experiments.} For each experiment, we fix $n=100$ agents and vary the number of items $m$ from $10^3$ to $10^5$. Each experimental point represents the average over multiple independent trials. Each agent draws utilities independently from a \textbf{distinct} distribution: in the \texttt{beta\_uniform} setting, each agent uses either a Beta$(\alpha,\beta)$ distribution (with parameters $\alpha, \beta$ sampled per item) or a Uniform$[a,b]$ distribution with parameters $a, b$ sampled per item); in the \texttt{normal\_uniform} setting, agents draw from either a truncated Normal$(\mu, \sigma^2)$ (with parameters $\mu, \sigma$ sampled per item) or a Uniform distribution (as described above).

We record three quantities: the \textbf{worst envy ratio} (the maximum over agents of each agent's envy relative to their own value), the \textbf{fraction of agents with any envy}, and the \textbf{welfare ratio} comparing sampled and non-sampled allocations.

\noindent
Envy behaviour. Across all distributions, envy drops rapidly as $m$ increases \mbox{(Figures~\ref{fig:goods-beta} and \ref{fig:goods-normal})}. Worst-envy ratios begin between roughly $4$-$9$ at $m = 10^3$ but approach $1$ by about $m = 5 \times 10^4$. The fraction of agents with any envy similarly falls to zero as $m$ grows. Sampling introduces the expected trade-off: small $s$ leads to higher envy at moderate $m$, but as soon as $s \geq 23 \approx 5 \ln n$, it becomes nearly indistinguishable from the non-sampling algorithm.

\noindent
Welfare under sampling. The welfare ratio between sampled and non-sampled allocations remains high throughout (Figures \mbox{\ref{fig:welfare-beta}} and \mbox{\ref{fig:welfare-normal})}. Even with $s=4$, the ratio stays above $0.85$, and for $s=23$ or $46$ it consistently exceeds $0.95$ across all values of $m$. The ratio is essentially flat in $m$, indicating that the efficiency impact depends primarily on $s$.

\begin{figure}[t]
    \centering
    \includegraphics[width=\picturewidth]{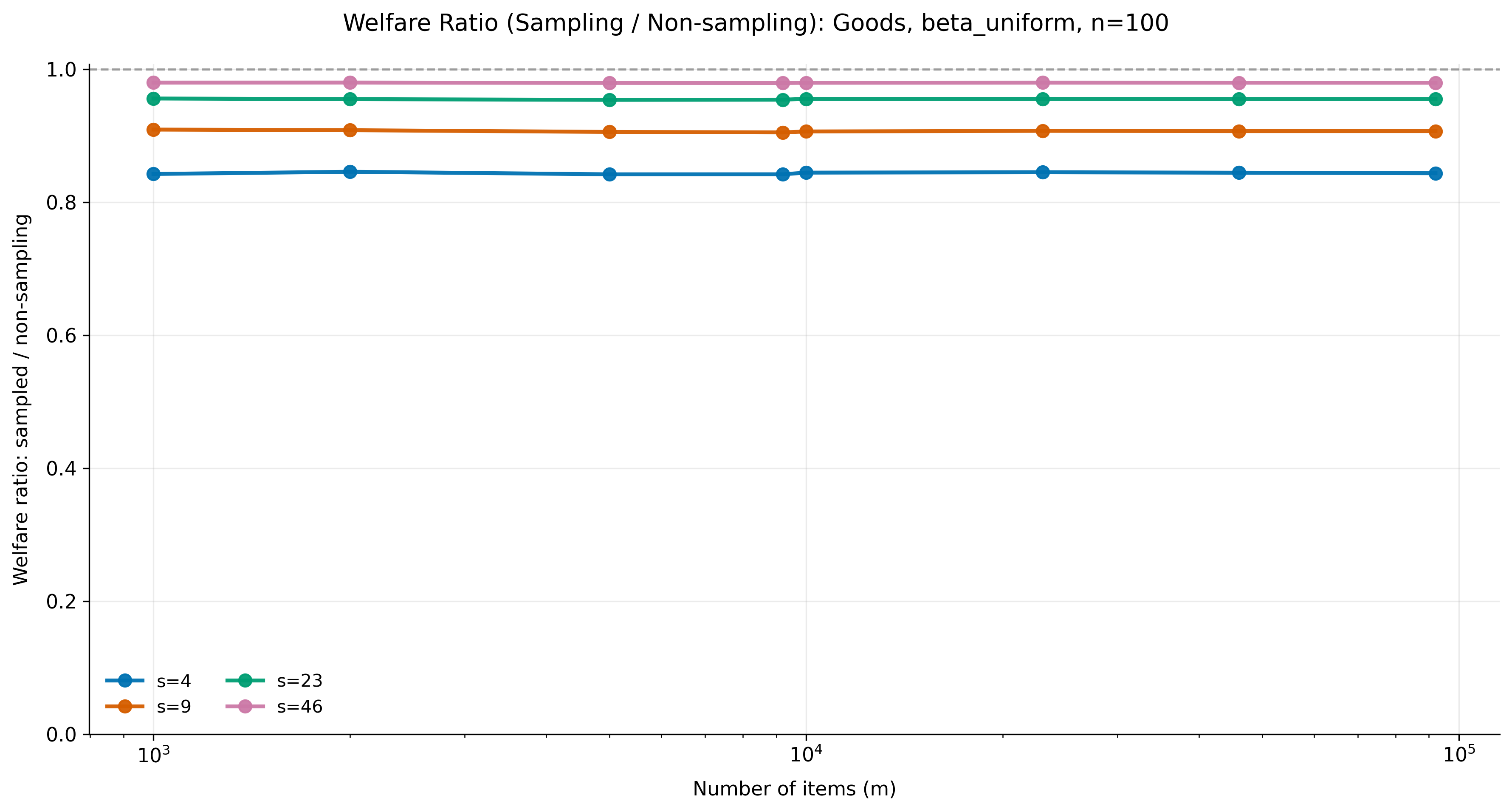}
    \caption{Welfare ratio (sampling / non-sampling): Goods, \texttt{beta\_uniform}.}
    \label{fig:welfare-beta}
    % \Description{Welfare ratio (sampling / non-sampling): Goods, \texttt{beta\_uniform}.}
\end{figure}

\begin{figure}[t]
    \centering
    \includegraphics[width=\picturewidth]{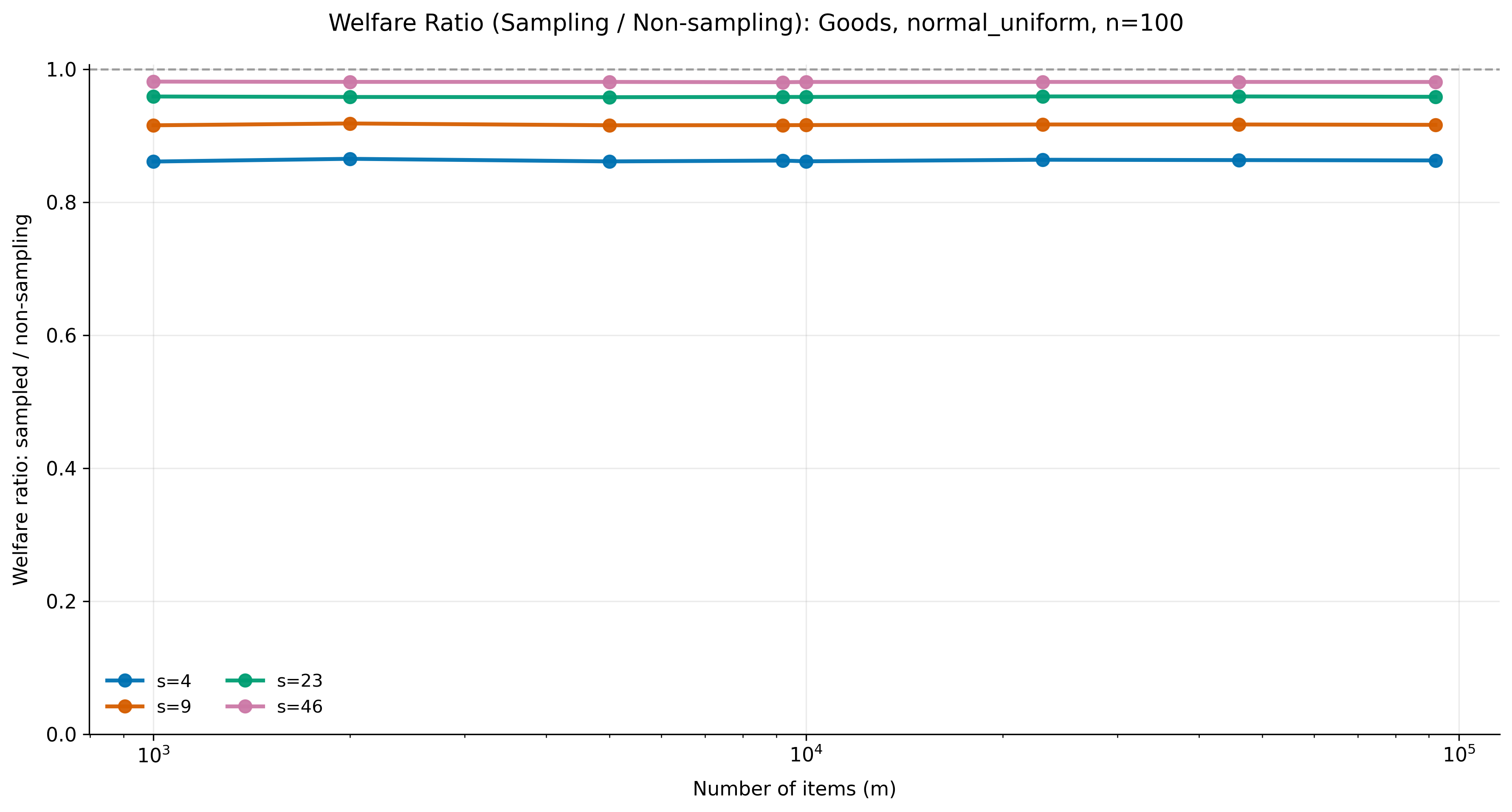}
    \caption{Welfare ratio (sampling / non-sampling): Goods, \texttt{normal\_uniform}.}
    \label{fig:welfare-normal}
    % \Description{Welfare ratio (sampling / non-sampling): Goods, \texttt{normal\_uniform}.}
\end{figure}

\noindent
\textit{Summary}.
Overall, the empirical results match the asymptotic predictions: large item counts rapidly eliminate envy, and sampling with moderate $s$ preserves both fairness and efficiency. The sampling algorithm thus offers a practical and scalable alternative with only minimal welfare loss.

%% file: 9.discussion-future.tex
\section{Discussion and Future Work}
\label{sec:Discussion}

This paper studies the problem of existence and efficient computation of envy-free or proportional fair, and maximum social welfare efficient allocations, in the setting of random instances. We extend known results by studying non-identical items. We also analyze efficient computation, and show that sampling a few utility values per item suffices to compute these allocations when they exist. Our results are in the asymptotic setting, that is, hold in the limit as the numbers of agents and items increase indefinitely. We complement these results by simulating our algorithms, and show that similar guarantees are obtained in practice for sufficiently large but finitely many agents and items.

%% file: 10.appendix.tex
\clearpage

\appendix
% START -----------------------------------------------------------

\section{Preliminaries}
In this section, we formally present our model for allocating chores. This is similar to the model for allocating goods, with slight notational modification for representing disutilities instead of utilities, and natural changes in the definitions of fairness and efficiency. We then state two probabilistic inequalities used in the proof of Theorem \ref{thm:prop-small-m}, and end with a summary of results proved in the rest of the Appendix, along with a description of its organization.

\paragraph{Model and Distribution(s) Over Disutilities (for Chores).} Let $N = \{1, \dots, n\}$ be the set of agents and $M = \{1, \dots, m\}$ be the set of chores. For each $(i, j) \in (N \times M)$, $d_i(j) \in [0, 1]$ denotes the \textit{disutility} of chore $j$ for agent $i$. The disutilities are assumed to be \textit{additive}, that is, for all $M' \subseteq M$, $d_i(M') = \sum_{j \in M'} d_i(j)$.

For each chore $j \in M$, there exists a distribution $\mathcal{D}_j$ supported on $[0,1]$, such that $d_i(j) \sim \mathcal{D}_j$ independently for each $i \in N$. The distributions have the following properties. Each $\mathcal{D}_j$ is \textit{non-atomic}, i.e., $\Pr_{X \sim \mathcal{D}_j}[X = x] = 0$ for all $x \in [0, 1]$, and also $(\alpha_j, \beta_j)$-PDF bounded. Here the $\alpha_j , \beta_j > 0$ may be distinct for each $j \in M$.

% Let $m = xn + y$ for some $x \in \mathbb{Z}_{\geq 0}$ and $0 \leq y < n$.

An allocation $A$ is \textit{envy-free} if for all pair of agents $i, i' \in N$, it holds that $d_i(A_i) \leq d_i(A_{i'})$. The envy of any agent $i \in N$ towards another agent $i' \in N$ is $\max \{0, d_i(A_{i}) - d_i(A_{i'})\}$. An allocation $A$ is \textit{proportional} if for all agents $i \in N$, it holds that $d_i(A_i) \leq \nicefrac{d_i(M)}{n}$.  Finally, an allocation $A$ has \textit{maximum social welfare} if it minimizes the sum of agent disutilities for their own shares. Formally, if $\mathcal{A}$ is the set of all allocations of $M,$ then $A\in argmin_{\mathcal{A}} \sum_i d_i(A_i)$.

% \paragraph{Additional Inequalities.} The proof of Theorem \ref{thm:prop-small-m} used two stronger probabilistic inequalities, which we state here.

% \begin{lemma}[Lyapunov's Moment Inequality, \cite{karrNotes} Corollary 4.110, p.182]\label{lem:Lyapunov's inequality}
%     Let \( X_j \) be a real-valued random variable with \( \mathbb{E}[X_j^2] < \infty \). Then,
%     \[
%     \mathbb{E}[|X_j|^3] \leq \left( \mathbb{E}[X_j^2] \right)^{3/2} = \sqrt{\mathbb{E}[X_j^2]} \cdot \mathbb{E}[X_j^2].
%     \]
% \end{lemma}

% \begin{lemma}[Berry-Esseen Inequality, ~\cite{esseen1942liapunoff}]\label{lem:Berry-Esseen}
%     Let \( X_1, \ldots, X_n \) be independent real-valued random variables with zero mean and finite third absolute moments. Let
%     $
%     S_n = \sum_{j=1}^n X_j
%     $
%     and
%     $
%     \sigma_n^2 = \sum_{j=1}^n \mathbb{E}[X_j^2] > 0.
%     $
%     Let \( F_n \) denote the distribution function of \( S_n / \sigma_n \), and let \( \Phi \) denote the standard normal distribution function. Then,
%     \[
%     \sup_x \left| F_n(x) - \Phi(x) \right| \leq C_0 \cdot \frac{ \sum_{j=1}^n \mathbb{E}[|X_j|^3] }{ \left( \sum_{j=1}^n \mathbb{E}[X_j^2] \right)^{3/2} }.
%     \]
%     where $C_0$ is fixed constant.
% \end{lemma}

\paragraph{Organization.} In the Appendix, we prove all the results stated in Section \ref{sec:intro-contribution} that have not been proven in the main section. In Section \ref{sec:ef-large-chores}, we prove a similar result for chores. Section \ref{sec:ef-small-chores} proves the existence of envy-free allocations for a small number of and chores. Finally, in Section \ref{sec:prop-large-m}, we prove the existence of proportional allocations for the intermediate case when the number of goods is linearly larger than the number of agents, but it is known that envy-free allocations do not exist for this case. All our results are also summarized in Tables \ref{tab:results-summary-non-sampling} and \ref{tab:results-summary-sampling}.

% \begin{table}[t]
% \centering
% \renewcommand{\arraystretch}{1.15}
% \begin{tabular}{|
%     >{\centering\arraybackslash}p{1.25cm}|
%     >{\centering\arraybackslash}p{1.5cm}|
%     >{\centering\arraybackslash}p{1.5cm}|
%     >{\centering\arraybackslash}p{2cm}|
% }
% \hline
% \textbf{Theorem} & \textbf{Setting} & \textbf{Range of $m$} & \textbf{Guarantees} \\
% \hline
% \textbf{2.1} & Goods (offline) & $m = \Omega(n \log n)$ & EF + MSW \\ \hline
% \textbf{C.1} & Chores (offline) & $m = \Omega(n \log n)$ & EF + MSW \\ \hline
% \textbf{2.2} & Goods (online, discrete) & $m = \Omega(n \log n)$ & EF + MSW \\ \hline
% \textbf{2.3} & Goods (online, continuous) & $m = \Omega(n \log n)$ & $0.8$-approx. EF + $0.9$-approx. MSW \\ \hline
% \textbf{2.4} & Goods (online, continuous, bounded means) & $m = \Omega(n \log n)$ & EF + $(1 - \frac{1}{\log n})$-approx. MSW \\ \hline
% \textbf{2.5} & Goods (offline) & $m \ge 5n$, \newline $m \equiv 0 \mod n$ & EF \\ \hline
% \textbf{2.6} & Chores (offline) & $m \ge 5n$ & EF \\ \hline
% \textbf{2.7} & Goods (offline) & $m = \Omega(n)$& PROP \\ \hline
% \textbf{2.8} & Goods (offline) & $n \le m \le 2n$ & PROP \\
% \hline
% \end{tabular}
% \caption{Summary of Theoretical Results}
% \label{tab:results-summary}
% \end{table}

\begin{table}[t]
\centering
\renewcommand{\arraystretch}{1.15}
\begin{tabular}{|
    >{\centering\arraybackslash}p{1.5cm}|
    >{\centering\arraybackslash}p{2.5cm}|
    >{\centering\arraybackslash}p{2cm}|
}
\hline
& \makebox{\textbf{Goods}} & \makebox{\textbf{Chores}} \\
\hline
\textbf{Large \# of items} & EF + MSW  (Thm \ref{thm:Goods-Large-m})& EF + MSW (Thm \ref{thm:Chores-Large-m})\\ \hline
\textbf{Small \# of items} & if $n\mid m:$ EF (Thm \ref{thm:Goods-Small-m}) else PROP (Thms \ref{thm:prop-large-m}, \ref{thm:prop-small-m}) & EF (Thm \ref{thm:Chores-Small-m}) \\ \hline
\end{tabular}
\caption{Summary of results without sampling. The large number of items setting corresponds to when $n=O(m/\log m),$ and the small number of items is the case when $m=O(n \log n).$ EF denotes envy-freeness, PROP denotes proportionality, and MSW denotes maximum social welfare.}
\label{tab:results-summary-non-sampling}
\end{table}

\begin{table}[t]
\centering
\renewcommand{\arraystretch}{1.15}
\begin{tabular}{|
    >{\centering\arraybackslash}p{2cm}|
    >{\centering\arraybackslash}p{3cm}|
    >{\centering\arraybackslash}p{2cm}|
}
\hline
\textbf{Distribution Model} & \makebox{\textbf{Guarantees}} & \makebox{\textbf{\# of Samples}} \\
\hline
Discrete $\mathcal{U}_j$ & \mline{EF+MSW (Thm \ref{thm:sampling-discrete})} & $s = \frac{2 \log{m}}{\alpha_{min}}$ \\ \hline
Continuous $\mathcal{U}_j$ & \mline{$0.8$-EF+$0.9$-MSW(Thm\ref{thm:sampling-continuous})} & $s = \frac{20 \log{m}}{\alpha_{min}}$ \\ \hline
Continuous $\mathcal{U}_j$ + Bounded Means & \mline{EF+$\left( 1 - \nicefrac{1}{\log{n}}\right)$-MSW(Thm\ref{thm:sampling-continuous-constant})}& $s = \frac{2 (\log{m})^2}{\alpha_{min}}$ \\ \hline
\end{tabular}
\caption{Summary of results with sampling for allocating goods when $m=O(n\log n)$. EF denotes envy-freeness and MSW denotes maximum social welfare.}
\label{tab:results-summary-sampling}
\end{table}

\section{Envy-freeness for Large number of Chores}\label{sec:ef-large-chores}

In a fashion very similar to the proof of Lemma \ref{lem:Utility-Gap}, we can also prove the following Lemma (proof omitted).

\begin{lemma}\label{lem:Disutility-Gap}
    Let $\mathcal{D}$ be a \textbf{non-atomic}, \textbf{continuous} distribution and $d_1, \dots, d_n$ are $n$ i.i.d. draws from $\mathcal{D}$, then for any $i, i' \in [n], i \neq i'$

    \begin{equation*}
        \mathbb{E}[d_i \mid \arg \min_{j \in [n]} d_j = i] < \mathbb{E}[d_i \mid \arg \min_{j \in [n]} d_j = i']
    \end{equation*}
\end{lemma}

An equivalent statement as theorem \ref{thm:Goods-Large-m} can be made for chores.

\begin{theorem}\label{thm:Chores-Large-m}
    When the number of chores is large, that is, $n = O(m/\log m),$ then, with probability $(1-\nicefrac{1}{m})$, an envy-free and maximum social welfare allocation exists as $m\rightarrow \infty.$
\end{theorem}

\begin{proof}
    Construct an allocation by giving each chore \( j \in M \) to the agent that has the least disutility from it: \( \arg \min_{k \in N} d_k(j) \). The allocation has maximum social welfare by construction. We can prove that this allocation is EF with high probability.

    For every chore \( j \in M \), each agent \( i \in N \) draws their disutility \( d_i(j) \sim \mathcal{D}_j \), where \( \mathcal{D}_j \) is a \textit{non-atomic}, \textit{continuous} distribution with support in \( [0, 1] \). By Lemma \ref{lem:Disutility-Gap} for each \( j \in M \), there exist constants \( \mu_j \) and \( \mu_j^* \) such that, for any \( i, i' \in N \):

    \begin{align*}
    & 0 < \mathbb{E} \left[d_i(j) \, \middle| \, \arg \min_{k \in N} d_k(j) = \{i\} \right] \leq \mu_j \\
    & < \mu_j^* \leq \mathbb{E} \left[d_i(j) \, \middle| \, \arg \min_{k \in N} d_k(j) = \{i'\} \right]
    \end{align*}

    Define for fixed \( i \in N, j \in M \)):

    \[
    X_i^j =
    \begin{cases}
    d_i(j), & \text{if } \arg \min_{k \in N} d_k(j) = \{i\} \\
    0, & \text{otherwise}
    \end{cases}
    \]

    Then \( d_i(A_i) = \sum_{j \in M} X_i^j \). Moreover,

    \begin{align*}
    \mathbb{E}[X_i^j] 
    & = \Pr[\arg \min_{k \in N} d_k(j) = \{i\}] \cdot \\ 
    & \mathbb{E}[d_i(j) \mid \arg \min_{k \in N} d_k(j) = \{i\}] \\
    & = \frac{1}{n} \cdot \mathbb{E}[d_i(j) \mid \arg \min_{k \in N} d_k(j) = \{i\}] \leq \frac{\mu_j}{n}
    \end{align*}

    Now, define for \( i, i' \in N \), \( i \neq i' \)), and \( j \in M \):

    \[
    Y_{ii'}^j =
    \begin{cases}
    d_i(j), & \text{if } \arg \min_{k \in N} d_k(j) = \{i'\} \\
    0, & \text{otherwise}
    \end{cases}
    \]

    Then \( d_i(A_{i'}) = \sum_{j \in M} Y_{ii'}^j \). Furthermore,

    \begin{align*}
    \mathbb{E}[Y_{ii'}^j]
    = \frac{1}{n} \cdot \mathbb{E}[d_i(j) \mid \arg \min_{k \in N} d_k(j) = \{i'\} ] \geq \frac{\mu_j^*}{n}
    \end{align*}

    Using the linearity of expectation:

    \begin{align*}
    & \mathbb{E}[d_i(A_{i'})] = \sum_{j \in M} \mathbb{E}[Y_{ii'}^j] \geq \mu_a^* \cdot \frac{m}{n}, \quad \\
    & \text{where } \mu_a^* = \frac{1}{m} \sum_{j \in M} \mu_j^*
    \end{align*}

    Let \( Z_i^j \in [0, 1] \) be such that \( \mathbb{E}[Z_i^j] = \frac{\mu_j}{n} \) and \( Z_i^j \) stochastically dominates \( X_i^j \). Then:

    \[
    \sum_{j \in M} \mathbb{E}[Z_i^j] = \mu_a \cdot \frac{m}{n}, \quad \text{where } \mu_a = \frac{1}{m} \sum_{j \in M} \mu_j
    \]

    and for all \( x \in \mathbb{R}^+ \):

    \[
    \Pr\left[\sum_{j \in M} Z_i^j \geq x\right] \geq \Pr\left[\sum_{j \in M} X_i^j \geq x\right]
    \]

    Let \( E_{ii'} \) be the event that agent \( i \) envies agent \( i' \), i.e., \( \sum_{j \in M} Y_{ii'}^j < \sum_{j \in M} X_i^j \). This happens only if

    \begin{align*}
    \sum_{j \in M} X_i^j 
    & \geq \mu_a \cdot \frac{m}{n} + \frac{\mu_a^* - \mu_a}{2} \cdot \frac{m}{n} \\
    & = \left(1 + \frac{\mu_a^* - \mu_a}{2\mu_a}\right) \mu_a \cdot \frac{m}{n} \\
    & = \left(1 + \frac{\mu_a^* - \mu_a}{2\mu_a}\right) \mathbb{E}\left[\sum_{j \in M} Z_i^j\right]
    \end{align*}

    or

    \begin{align*}
    \sum_{j \in M} Y_{ii'}^j 
    & \leq \mu_a^* \cdot \frac{m}{n} - \frac{\mu_a^* - \mu_a}{2} \cdot \frac{m}{n} \\
    & = \left(1 - \frac{\mu_a^* - \mu_a}{2\mu_a^*}\right) \mu_a^* \cdot \frac{m}{n} \\
    & \leq \left(1 - \frac{\mu_a^* - \mu_a}{2\mu_a^*}\right) \mathbb{E}\left[\sum_{j \in M} Y_{ii'}^j\right]
    \end{align*}

    Let \( \epsilon = \min\left\{1, \frac{\mu_a^* - \mu_a}{2\mu_a^*}\right\} \), so \( \epsilon < \frac{\mu_a^* - \mu_a}{2\mu_a} \) also holds. $\epsilon > 0$ also holds as $\mu_j < \mu_j^*$ for all $j$ implies $\mu_a < \mu_a^*$.

    Since \( Z_i^j \) and \( Y_{ii'}^j \) are independent across $j$, using Chernoff bounds:

    \begin{align*}
    & \Pr\left[ \sum_{j \in M} Y_{ii'}^j \leq (1 - \epsilon) \cdot \mathbb{E}\left[ \sum_{j \in M} Y_{ii'}^j \right] \right]\\ 
    & \leq \exp\left( - \frac{\epsilon^2}{2} \mu_a^* \cdot \frac{m}{n} \right)
    \end{align*}

    and

    \begin{align*}
    & \Pr\left[ \sum_{j \in M} X_i^j \geq (1 + \epsilon) \cdot \mathbb{E}\left[ \sum_{j \in M} Z_i^j \right] \right] \\
    & \leq \exp\left( - \frac{\epsilon^2}{3} \mu_a \cdot \frac{m}{n} \right)
    \end{align*}

    Set \( n \leq \frac{\epsilon^2 \mu_a}{3} \cdot \frac{m}{\ln(2m^3)} \). By the union bound:

    \begin{align*}
    \Pr[E_{ii'}]
    &\leq \exp\left( - \frac{\epsilon^2}{2} \mu_a^* \cdot \frac{m}{n} \right) + \exp\left( - \frac{\epsilon^2}{3} \mu_a \cdot \frac{m}{n} \right) \\
    &\leq 2 \cdot \frac{1}{2m^3} = \frac{1}{m^3}
    \end{align*}

    Allocation \( A \) is EF if and only if no event \( E_{ii'} \) occurs. The probability that \( A \) is not EF is at most:

    \[
    \Pr\left[ \bigvee_{\substack{i, i' \in N \\ i \neq i'}} E_{ii'} \right]
    \leq \sum_{\substack{i, i' \in N \\ i \neq i'}} \Pr[E_{ii'}]
    \leq \binom{n}{2} \cdot \frac{1}{m^3} \leq \frac{1}{m}
    \]

    Thus, the probability that \( A \) is not EF vanishes as \( m \to \infty \) (\( n \to \infty\) and \(m = \Omega(n \log{n})\)).
\end{proof}

\section{Envy-freeness for Small Number of Chores}\label{sec:ef-small-chores}

A stronger statement than \ref{thm:Goods-Small-m} can be made for chores. For identical chores, \cite{manurangsi2025asymptoticfairdivisionchores} show that an envy-free solution exists with high probability when $m \ge 2n.$ We generalize this to the non-identical case.

\begin{restatable}[Appendix \ref{sec:ef-small-chores}]{theorem}{efsmallchores}
\label{thm:Chores-Small-m}
For allocating chores, if $m \ge 5n,$ then with high probability, an envy-free allocation exists as $n\rightarrow \infty$.
\end{restatable}

% \begin{theorem}
%     If $m \geq 5n$, then with high probability, there exists an envy-free allocation of chores.
% \end{theorem}

\begin{proof}
    Let $m = xn + y$ for some $x \in \mathbb{Z}_{\geq 0}$ and $0 \leq y < n$. Let $\alpha_{min} = \min_{j \in M} \alpha_j$, $\alpha_{max} = \max_{j \in M} \alpha_j$ and $\beta_{max} = \max_{j \in M} \beta_j$.

    Construct a bipartite graph $G = (N, M_{\leq nx}, E)$ where edges are given by:
    \[
    E = \left\{ (i, j) \in N \times M_{\leq nx} : d_i(j) \leq \tau_j \right\}
    \]
    where $M_{\leq nx} := \{j \in  M: j \leq nx\}$ (the first $nx$ goods in M), $\tau_j := \frac{1.1 \log{n}}{\alpha_j n}$ and $\tau := \max_{j \in M} \tau_j = \frac{1.1 \log{n}}{\alpha_{min} n}$.

    Let $\tau' := \frac{1}{n^{4/x + \varepsilon}}$ for some $\varepsilon > 0$ (say $\epsilon = 0.1$) and $x \geq 5$. The graph $G$ is an Erdős-Rényi random graph with edge probabilities at least $\frac{\log{n} + \omega(1)}{n}$, so by Lemma \ref{lem:Perfect-r-Matching} with high probability, there exists a perfect $x$-matching corresponding to an allocation $\{A_1^0, \dots, A_n^0\}$ such that:
    \[
    d_i(A_i^0) \leq x\tau
    \]

    Now fix any pair $i, i' \in N$. The probability that more than $\frac{x}{2}$ chores in $A_i^0$ have $d_{i'}(j) < \tau'$ is bounded by:
    \begin{align*}
    \Pr\left( |\{j \in A_i^0 :d_{i'}(j) < \tau'\}| > \frac{x}{2} \right) 
    & \leq \left(\frac{\beta_{max}}{n^{4/x + \varepsilon}} \right)^{x/2} \\ 
    & = o(n^{-2})
    \end{align*}

    By union bound over all $i, i' \in N$:
    \begin{align*}
    & \Pr \left( \exists i,i' \in N: |\{j \in A_i^0 : d_{i'}(j) < \tau'\}| > \frac{x}{2} \right) \\ 
    & = n^2 \cdot o(n^{-2}) = o(1)
    \end{align*}

    Hence, with high probability:
    \[
    d_i(A_{i'}^0) \geq \frac{x}{2} \tau'
    \Rightarrow d_i(A_i^0) - d_i(A_{i'}^0) \leq x \tau - \frac{x}{2} \tau'
    \]
    \[
    = \frac{x}{2} (2 \tau - \tau') = \frac{x}{2} \left(\frac{2.2 \log{n}}{\alpha_{min} n} - \frac{1}{n^{4/x + \varepsilon}} \right)
    \]

    We want this to be at most $- \frac{1.1 \log{n}}{\alpha_{min} n}$. Thus,
    \[
    \frac{x}{2} \left( \frac{1}{n^{4/x + \varepsilon}} - \frac{2.2 \log{n}}{\alpha_{min} n} \right) \geq \frac{1.1 \log{n}}{\alpha_{min} n}
    \]
    
    For sufficiently large $n$, $\frac{1}{2 n^{4/x + \varepsilon}} \geq \frac{2.2 \log{n}}{\alpha_{min} n}$. Hence, it is sufficient that
    \[
    \frac{x}{4} \cdot \frac{1}{n^{4/x + \varepsilon}} \geq \frac{1.1 \log{n}}{\alpha_{min} n}
    \Rightarrow x \geq \frac{4.4 \log{n}}{\alpha_{min} n^{1 - 4/x - \varepsilon}}
    \]

    Therefore, it suffices to choose $x \geq 5$ (inequality holds for sufficiently large $n$: $\frac{\log{n}}{n^{1 - \frac{4}{x} - \epsilon}}$ gets increasingly smaller with $n$ as $1 - \frac{4}{x} - \epsilon > 0$)

    Now construct a second graph $G' = (N, M \setminus M_{\leq nx}, E')$ where:
    \[
    E' = \left\{ (i, j) \in N \times \left( M \setminus M_{\leq nx} \right) : d_i(j) \leq \frac{1.1 \log{n}}{\alpha_{min} n} \right\}
    \]
    Then $G'$ is again an Erdős-Rényi graph with edge probabilities at least $\frac{\log{n} + \omega(1)}{n}$, so by Lemma \ref{lem:Right-Saturated-Matching} with high probability, there exists a right-saturated matching corresponding to allocation $\{A_1^1, \dots, A_n^1\}$.

    The final allocation $(A_1^0 \cup A_1^1, \dots, A_n^0 \cup A_n^1)$ is envy-free:

    For any $i, i' \in N$,
    \begin{align*}
        d_i(A_i) 
        & = d_i(A_i^0) + d_i(A_i^1) \\
        & \leq d_i(A_{i'}^0) - \frac{1.1 \log{n}}{\alpha_{min} n} + d_i(A_i^1) \\
        & \leq d_i(A_{i'}^0) - \frac{1.1 \log{n}}{\alpha_{min} n} + \frac{1.1 \log{n}}{\alpha_{min} n} \\
        & \leq d_i(A_{i'}^0) \leq d_i(A_{i'}^0) + d_i(A_{i'}^1) = d_i(A_{i'})
    \end{align*}
\end{proof}

\section{Proportionality for Linearly Large Number of Goods}\label{sec:prop-large-m}

%\subsection{Large number of goods $(m = \Omega(n))$}

\proplargegoods*

% \begin{theorem}
%     Let $c < 1$ be such that $\mathbb{E}[\mathcal{U}_j] \leq c$ for all $j \in M$. If $r = \lceil 2 \cdot \frac{3 + c}{1 - c} \rceil$ and $m \geq rn$, then with high probability, there exists a proportional allocation.
% \end{theorem}

% Theorem \ref{thm:Goods-Small-m} proves stronger results, in a similar fashion as the proof of the above theorem. We still leave the proof here for completeness.

\begin{proof}
    For this proof, we write $m$ as $m = xn+y$ for some $x \in \mathbb{Z}_{\geq 0}$ and $0 \leq y < n$. Let $\alpha_{min} = \min_{j \in M} \alpha_j$. Now, we define the bipartite graph $G = (N, M, E)$, where edges are determined as follows:
    \begin{equation*}
    E = \{ (i, j) \in N \times M : u_i(j) \geq \tau_j \}
    \end{equation*}
    where $\tau_j := 1 - \frac{1.1 \log n}{\alpha_j n}$, and define:
    \begin{equation*}
    \tau := \min_{j \in M} \tau_j = 1 - \frac{1.1 \log n}{\alpha_{min} n}
    \end{equation*}

    With high probability, there exists a left-saturated $x$-matching (each agent in $N$ is matched to exactly $x$ goods in $M$) in $G$ (since $G$ is an Erd\H{o}s-R\'enyi random graph with $p_{\min} = \frac{\log n + \omega(1)}{n}$), which corresponds to the allocation ${A_1^0, \dots, A_n^0}$ such that:
    \begin{equation*}
    u_i(A_i^0) \geq x\tau = x \left( 1 - \frac{1.1 \log n}{\alpha_{min} n} \right)
    \end{equation*}
    
    Let $\tau' := \frac{1 + c}{2}$. Using Markov's inequality:
    \begin{equation*}
    \Pr[u_i(j) > \tau'] \leq \frac{c}{\tau'} = \frac{2c}{1 + c} < 1
    \end{equation*}
    
    Thus, for any agent $i$, the probability of more than $\frac{m}{2}$ values exceeding $\tau'$ is exponentially small:
    \begin{align*}
    \Pr & \left( |\{j \in M: u_i(j) > \tau'\}| > \frac{m}{2} \right) 
    \leq \left( \frac{2c}{1 + c} \right)^{\frac{m}{2}} \\ 
    & = \left( 1 - \frac{1 - c}{1 + c}\right)^{\frac{m}{2}} \\
    & \leq \exp \left( - \frac{1 - c}{1 + c} \cdot \frac{m}{2} \right) = \exp \left( - \Theta(m) \right)
    \end{align*}

    By a union bound over all agents,
    \begin{align*}
    \Pr & \left( \exists i \in N \text{ s.t. } |\{j \in M: u_i(j) > \tau'\}| > \frac{m}{2} \right) \\
    & = n \exp \left( - \Theta(m) \right) \\
    & \leq m \exp \left( - \Theta(m) \right) = o(1)
    \end{align*}
    
    Let us bound $u_i(M)$. Each agent receives at most $\frac{m}{2}$ items with value at most $\tau'$, and the rest (at most $\frac{m}{2}$) are upper bounded by $1$. Hence:
    \begin{equation*}
    u_i(M) \leq \frac{m}{2} \cdot \tau' + \frac{m}{2} \cdot 1 = \frac{m(1 + \tau')}{2} = \frac{m(3 + c)}{4}
    \end{equation*}
    
    Therefore:
    \begin{equation*}
    \frac{u_i(M)}{n} \leq \frac{(x + 1)(3 + c)}{4} = x\cdot \frac{3 + c}{4} + \frac{3 + c}{4}
    \end{equation*}
    
    We now compute:
    \begin{align*}
    x \tau - \frac{u_i(M)}{n} 
    & \geq x \left(1 - \frac{1.1 \log n}{\alpha_{min} n}\right) - x\cdot \frac{3 + c}{4} - \frac{3 + c}{4} \\ 
    & = x \left( \frac{1 - c}{4} - \frac{1.1 \log n}{\alpha_{min} n} \right) - \frac{3 + c}{4}
    \end{align*}

    \begin{equation*}
    x \left( \frac{1 - c}{4} - \frac{1.1 \log n}{\alpha_{min} n} \right) - \frac{3 + c}{4} \geq 0 \Rightarrow  x \tau - \frac{u_i(M)}{n} \geq 0
    \end{equation*}
    
    Hence, we want:
    \begin{align*}
    & x \left( \frac{1 - c}{4} - \frac{1.1 \log n}{\alpha_{min} n} \right) \geq \frac{3 + c}{4} \\  
    & \Leftrightarrow x \geq \frac{3 + c}{(1 - c) - \frac{4.4 \log n}{\alpha_{min} n}}
    \end{align*}
    
    For sufficiently large $n$ such that $\frac{4.4 \log n}{\alpha_{min} n} \leq \frac{1 - c}{2}$, we have
    \begin{equation*}
    x \geq 2 \cdot \frac{3 + c}{1 - c}
    \end{equation*}
    
    Hence, it suffices to set:
    \begin{align*}
    & r := \left\lceil 2 \cdot \frac{3 + c}{1 - c} \right\rceil \quad \text{and} \quad x \geq r
    \end{align*}

    Let $(A_1^1, \dots, A_n^1)$ be any allocation of the remaining $y$ goods. The allocation $(A_1^0 \cup A_1^1, \dots, A_n^0 \cup A_n^1)$ is proportional with high probability:
    \begin{equation*}
    u_i(A_i^0 \cup A_i^1) \geq u_i(A_i^0) \geq \frac{u_i(M)}{n}, \quad \forall i \in N
    \end{equation*}
\end{proof}

\section{Details on Empirical Results}\label{sec:details-experiments}

This section provides the full experimental setup corresponding to the empirical results in Section~\ref{sec:empirical-results}. All code used for generating instances, running allocations, and producing figures is available at 

\begin{center}
    \textcolor{blue}{~\url{https://github.com/AprupKale/Aysmptotic-Fair-Division-Simulations}}.
\end{center}

\subsection{Algorithms Evaluated}
We compare two allocation procedures implemented in the repository:
\begin{itemize}
    \item \textbf{Non-sampling algorithm}: the baseline algorithm that uses each item's full utility vector over all $n$ agents.
    \item \textbf{Sampling algorithm}: for each item $j$, we randomly sample $s$ agents (without replacement) and allocate to one of these agents (to one with the highest utility / lowest disutility).
\end{itemize}

\subsection{Utility Distribution}
Each item's utility distribution is distinct and independently sampled for each agent. Several heterogeneous families of distributions are considered. The two families, whose results are presented are:

\paragraph{\texttt{beta\_uniform} mixture.}
For each item $j$, and for all agents $i$, $u_{i}(j)$ is drawn from:
\begin{itemize}
    \item a Beta$(\alpha,\beta)$ distribution, with $\alpha,\beta$ \textbf{resampled independently for every item}, or
    \item a Uniform$[a,b]$ distribution, with $a,b$ likewise resampled per item.
\end{itemize}

\paragraph{\texttt{normal\_uniform} mixture.}
For each item $j$, and for all agents $i$, $u_{i}(j)$ is drawn from:
\begin{itemize}
    \item a truncated (values are truncated to $[0,1]$ and then normalized so that $\Pr([0, 1]) = 1$ for the distribution) Normal$(\mu,\sigma^2)$ value, with $(\mu,\sigma)$ resampled independently per item, or
    \item a Uniform$[a,b]$ value as above.
\end{itemize}

All the distributions are implemented in \texttt{dists.py}.

\subsection{Multiple Trials and Averaging}
For each $(n, m, s)$ triple, we repeat the experiment across multiple seeded trials. Seeds are explicitly controlled in the repository to guarantee reproducibility across algorithms. Each plotted point in Section~\ref{sec:empirical-results} represents the \textbf{mean} over trials; the code also computes standard deviations, though we omit error bars for figure clarity.

\subsection{Metrics Computed}
Each trial produces the following metrics:

\begin{itemize}
    \item \textbf{Worst envy ratio}: for each agent $i$, compute the maximum envy $E_i = \max_{k \neq i} u_i(A_k) - u_i(A_i)$ when positive, normalized by $u_i(A_i)$; the reported value is $\max_i (1 + \nicefrac{E_i}{u_i(A_i)})$.
    \item \textbf{Fraction of agents with envy}: the fraction of agents for whom $E_i > 0$.
    \item \textbf{Welfare ratio}: the total welfare of the sampled allocation divided by the welfare of the non-sampled allocation for the same trial.
\end{itemize}

All metrics are computed using the same utility matrix to ensure direct comparability across algorithms.

\subsection{Figure Generation}
All plots are generated using the scripts in \texttt{visualize\_results.py}, which load cached experiment results if they already exist. To avoid recomputation, each trial is stored as a \texttt{json} file indexed by $(n, m, s)$ and distribution type. 

% END -------------------------------------------------------------